%% file: onedimwordstrees.tex
\documentclass{lmcs}

\usepackage{amssymb}
\usepackage{tikz}
\usetikzlibrary{calc,decorations.pathmorphing,snakes}

\usepackage{hyperref}
\usepackage{verbatim}

\theoremstyle{plain}\newtheorem{claim}[thm]{Claim}

\def\eg{{\em e.g.}}

\def\ie{{\em i.e.}}
\def\etc{{\em etc.}}

\input{macros.tex}

\begin{document}

\title{One-Dimensional Fragment over Words and Trees}
\titlecomment{{\lsuper*}This is an extended and revised version of \cite{Kie16} and \cite{KK17}.}

\author[E.~Kiero\'nski]{Emanuel Kiero\'nski}	
\address{Institute of Computer Science, University of Wroc\l{}aw}	
\email{emanuel.kieronski@cs.uni.wroc.pl}  

\author[A.~Kuusisto]{Antti Kuusisto}	
\address{University of Helsinki and Tampere University, Finland}	
\email{antti.kuusisto@helsinki.fi}  

\begin{abstract}
  \noindent 
	One-dimensional fragment of first-order logic is obtained by restricting quantification to blocks of existential (universal) quantifiers that leave at most
one variable free. We investigate this fragment over words and trees, presenting a complete classification of the complexity of its satisfiability problem
for various navigational signatures, and comparing its expressive power with other important formalisms.
These include the two-variable fragment with counting
and the unary negation fragment.
\end{abstract}

\maketitle

\section{Introduction}\label{S:one}

One-dimensional fragment of first-order logic, \ODF{}, is obtained by restricting quantification to blocks of existential quantifiers that leave at most
one variable free. As the logic is closed under negation, one may also use blocks of universal quantifiers.  
\ODF{} contains a few known decidable fragments of first-order logic: the prenex form class $\forall \exists^*$ with equality, the two-variable fragment \FOt, and (the so-called UN-normal form of formulas
in) the unary negation fragment \UNFO{}. 

Unfortunately,  over general relational structures, the satisfiability problem for \ODF{} is undecidable \cite{HK14}.
In such situation, one may attempt to regain the decidability in two principal ways: by imposing some additional restrictions on the syntax of  the considered
logic or by restricting attention to some specific classes of structures.

Regarding the first idea, a nice syntactic restriction of \ODF{}, which turns out to be decidable over the class of all relational structures, is called the \emph{uniform} one-dimensional fragment \UF{}. It
 was introduced by Hella and Kuusisto in \cite{HK14} as a generalization of the two-variable fragment of first-order logic to contexts with relations of all arities---in particular, relations with arities greater than two.
Such contexts naturally include, e.g., databases.
The readers interested in this variant are referred to \cite{HK14}, \cite{KK14}, \cite{KK15} and the survey \cite{KUS16}, the latter also revealing some connections to
description logics.

In this paper we will investigate \ODF{} over restricted classes of structures. There are two important options, well motivated in various areas of computer science, namely, the  class of words and the class of trees. 
Our aim is to investigate the complexity of the satisfiability problem of \ODF{} over words and trees, and to compare the
expressive power of \ODF{} over these classes of structures with a few other formalisms considered in this context. 
To set up the scene, let us recall the main results on satisfiability of fragments of first-order logic over words and trees.

Over words, it is known that the satisfiability problem for full first-order logic
is decidable, but with non-elementary complexity. In fact,  as shown by Stockmeyer \cite{Sto74}, already the fragment with three variables
is non-elementary. On the other hand, a reasonable complexity is obtained when the number of variables is restricted to two. The satisfiability problem for \FOt{} 
over words and $\omega$-words was shown to be \NExpTime-complete by Etessami, Vardi and Wilke \cite{EVW02}. In the same paper it was observed that
the expressive power of \FOt{} over words is equal to the expressive power of unary temporal logic, \UTL, \ie,~linear temporal logic with the four navigational operators \emph{next state},
\emph{somewhere in the future}, \emph{previous state}, \emph{somewhere in the past}. \FOt{}, however, turns out to be exponentially more succinct than \UTL.
The extension 
of \FOt{} by counting quantifiers, \Ct{},  was shown to be  \NExpTime-complete over words by Charatonik and Witkowski \cite{CW16}. In fact, it is 
not difficult to observe that over words, \Ct{} has the same expressive power as plain \FOt{}.
Another interesting 
extension of \FOt{}, this time significantly increasing its expressive power, is the extension by the \emph{between} predicate recently studied by Krebs et al.~\cite{KLP20}. Satisfiability for this logic is \ExpSpace-complete.

Turning then to the class of trees,
both \FOt{} and \Ct{} retain a reasonable complexity, namely,
their satisfiability problems  over trees are \ExpSpace-complete. See Benaim et al.~\cite{BBC16}
for the analysis of \FOt{} over trees and Bednarczyk, Charatonik and Kiero\'nski \cite{BCK17} for an extension covering
\Ct{}. Regarding the expressive power, the situation depends on the type of trees
considered. In the case of \emph{unordered} trees, \FOt{} cannot count and is thus
less expressive than \Ct{}. Over \emph{ordered} trees, both formalisms are
equally expressive \cite{BCK17} and share the expressiveness with the navigational core 
of XPath, \CoreXPath{} (cf.~Marx and de Rijke \cite{MR04}), a logic similar in spirit to \UTL{}, 
used to reason about XML trees.

\medskip\noindent
{\bf Our results over words.} We first analyse the  expressive power and the complexity of the satisfiability problem of \ODF{} over words and $\omega$-words. In our scenario
we assume that at each position of a word ($\omega$-word), multiple unary predicates may be true, and two navigational binary predicates
are used to navigate structures: successor $\succh$ and its transitive closure $\lessh$. 
We show that the expressive power of \ODF{} over such structures is the same as the expressive power of \FOt{}, and thus also of \UTL{} and \Ct{}.

The advantage of \ODF{} over these other formalisms is that it allows to specify many properties in a more natural and elegant way. 
If we want to say that a word contains some (especially not fully specified) pattern, consisting of more than two elements, we can just quantify the appropriate number 
of positions and say how they should be labelled and related to each other. Expressing the same in \FOt{} will usually
require some heavy recycling of the two available variables.
Let us look at two simple examples. Consider a system whose behaviour we model as a word, or an $\omega$-word, in which one or more of the atomic propositions $P_1, \ldots, P_n$ can hold in a given point of time. 
To say that there are $m$ non-overlapping time intervals (sets of consecutive positions of the word) in each of which each $P_i$ holds at least once,
we can use the following \ODF$[\lessh]$ sentence:

\begin{align}
\exists y_0 y_1 \ldots y_n x_{11} \ldots x_{1n} \ldots x_{m1} \ldots x_{mn} (\bigwedge_{i=1}^m \bigwedge_{j=1}^n  y_{i-1} \lessh x_{ij} \wedge x_{ij} \lessh y_i \wedge P_jx_{ij}).
\end{align}

As another example\footnote{Suggested to the authors by Jakub Michaliszyn.}, not using the navigational predicates at all, consider the property saying that it is possible to choose $m$ positions satisfying together all of the $P_i$:

\begin{align}
\exists x_1 \ldots x_m  (\bigwedge_{i=1}^{n} \bigvee_{j=1}^m P_ix_j).
\label{e:two}
\end{align}

The reader can check that expressing the above properties in \FOt{}$[\succh, \lessh]$ is indeed not straightforward and leads to complicated formulas.

In fact, our translation of \ODF$[\succh, \lessh]$ to \FOt$[\succh, \lessh]$ has an exponential blow-up, which seems to be hard to avoid, and which thus suggests
that \ODF$[\succh, \lessh]$ may be able to express some properties  more succinctly than \FOt$[\succh, \lessh]$, and possibly even
more succintly than \Ct$[\succh, \lessh]$.

Regarding the complexity, we show that satisfiability of \ODF{} over words and $\omega$-words is \NExpTime-complete, that is, it is of the same complexity
as satisfiability of \FOt{} and \Ct{}. While our proof has some similarities to the proof of Etessami, Vardi and Wilke \cite{EVW02} for \FOt$[\succh, \lessh]$, it is technically more difficult, due to
the combinatorically more complicated nature of the objects involved. Not surprisingly, the basic idea in the proof is  based
on an appropriately tuned contraction procedure.

We also examine some possible extensions of \ODW{}. perhaps the most significant of them is 
the extension of \ODW{} by an equivalence relation, inspired by an analogous extension of \FOt$[\succh, \lessh]$  (\FOt{} over \emph{data words}),
studied by Boja\'nczyk et al.~\cite{BDM11}. The satisfiability problem for \FOt{} over data words, even though very hard, is decidable.
We show that  \ODW{} over data words becomes undecidable.

\medskip\noindent
{\bf Our results over trees.}
 We consider finite unranked trees accessible by  navigational 
signatures built out of (some of) the following relations: child $\succv$, descendant $\lessv$, 
next sibling  $\succh$ and following sibling  $\lessh$.
Concerning the complexity of satisfiability, it turns out that it depends on whether $\succv$ is present or not.
With $\succv$ the satisfiability problem is \TwoExpTime-complete, and without $\succv$ it is
\ExpSpace-complete. To show the upper complexity bound in the case of the full navigational signature, we will use
the existing results for \UNFO{} by ten Cate and Segoufin \cite{SC13}. For the \ExpSpace{} bound we perform some surgery on
models leading to small model properties, and then design an algorithm
searching for such appropriate small models. Technically, we extend
 the approach from \cite{CKM13} used there in the context of \FOt{}. Roughly speaking,
we appropriately abstract the information about a node by its \emph{profile} (an analogous notion is called \emph{a full type} in \cite{CKM13}) 
and then we contract trees, removing their fragments between nodes with the same profiles. 
We explain also how to use these techniques to directly reprove the upper bound for the full signature.
The lower bounds are inherited from other formalisms.

It is worth mentioning that an orthogonal extension of the method from \cite{CKM13} is used
in \cite{BCK17} in the context of \Ct{}. In both cases the challenge is to carefully tune the notion of
a profile (full type) in order to get the optimal complexity.

Regarding expressivity, we show that over
ordered trees with all of the four navigational relations we consider, \ODF{} is
expressively equivalent to each of \CoreXPath, \GFt, \FOt{}, \Ct, \UNFO{}. 
We also show that over unordered trees equipped with both the descendant and the child relation, 
\ODF{} is still equivalent to \Ct{}, but we establish that this time  \FOt{} is less expressive, and that
\CoreXPath, \GFt{} and \UNFO{} are less expressive than \FOt{} (and equiexpressive with each other).
Most of these expressivity results
are rather easy to obtain (though in some cases slightly awkward to formally show).
The exception is the equivalence of \ODF{} and \Ct{} in the absence of the sibling relations,
which is  less obvious and more difficult to prove.
In our expressivity-related studies, we do not consider the cases of unordered trees accessible by only one of the descendant and the child relations.

\medskip\noindent
{\bf Organization of the paper.}
The rest of the paper is organized as follows. In Section \ref{s:prelim} we define the logics and structures we are interested in,  introduce some basic notions and
results which will then be used
in the following sections. In Section \ref{s:expwords} we compare the expressivity of \ODF{} with other formalisms over words and $\omega$-words,
in Section \ref{s:satwords} we analyse the complexity of \ODF{} over words and $\omega$-words, Section \ref{s:exptrees} concerns the expressive power of \ODF{} over trees,
and in Section \ref{s:sattrees} we analyse the complexity of \ODF{} over trees. Finally, in Section \ref{s:conc}, we conclude the paper.

\section{Preliminaries} \label{s:prelim}

\subsection{Structures}
We employ conventional terminology and notation from model theory throughout this article, assuming the reader is familiar with most
of the standard concepts.
We refer to structures using Gothic capital letters (e.g., $\str{M}$),
and their domains using the corresponding Roman capitals (e.g., $M$).

We are interested in signatures of the form $\sigma=\sigma_0 \cup \sigma_{nav}$, where $\sigma_0$ consists of some number of unary relation symbols,  and $\sigma_{nav}$, called the \emph{navigational signature}, is a subset of $\{\succh, \lessh, \succv, \lessv \}$. 

A \emph{word} is a finite structure over a signature $\sigma_0 \cup \{\succh, \lessh \}$ in which $\lessh$ is a (strict) linear order and $\succh$ its induced successor relation.
An infinite structure over the same signature
and containing a reduct isomorphic to $(\N, +1, < )$ is called an $\omega$-\emph{word}. Given a word $\str{M}$, its element $a$ and a number $i \in \N$, we will
sometimes refer by $a+i$ (respectively, $a - i$) to the element located $i$ positions to the right (resp., left) from $a$.  
We will also use the notation  $\str{M}=\str{M}_1a$ to denote that the word $\str{M}$ is the concatenation of the word $\str{M}_1$ with 
the element $a$. In the similar vein we will can write $\str{M}=\str{M}_1a\str{M}_2$, \etc

Let $\N^*$ denote the set of finite sequences of natural numbers, containing in particular the empty sequence $\epsilon$. For $\alpha, \beta \in \N^*$ and $i \in \N$, we denote by $\langle \alpha, i\rangle$ the sequence obtained as the result of 
appending $i$ to $\alpha$, and by $\langle \alpha, \beta\rangle$ the result of concatenating $\alpha$ and $\beta$.
A \emph{tree} is a finite structure $\str{T}$ whose  universe $T$ is a subset of $\N^*$ such that 
if $\langle \alpha, i\rangle \in T$, then $\alpha \in T$, and in the case  $i>0$, 
also $\langle \alpha, i-1\rangle \in T$. In a tree, at least one of $\succv$, $\lessv$ and possibly one or both of $\succh$, $\lessh$
are interpreted, each of them in the following fixed way. For $a,b \in T$, we have 
$\str{T} \models a \succv b$ iff $a=\alpha$ and $b=\langle \alpha, i\rangle$ for some $\alpha \in \N^*$ and $i \in \N$; 
$\str{T} \models a \lessv b$ iff $a=\alpha$ and $b=\langle \alpha, \beta\rangle$ for some $\alpha, \beta \in \N^*$, $\beta \not= \epsilon$; 
$\str{T} \models a \succh b$ iff $a=\langle \alpha, i\rangle$ and $b=\langle \alpha, i+1\rangle$ for some $\alpha \in \N^*$ and $i \in \N$; 
$\str{T} \models a \lessh b$ iff $a=\langle \alpha, i\rangle$ and $b=\langle \alpha, j\rangle$ for
some $\alpha\in \N^*$ and $i,j \in \N$, $i < j$.

When speaking about trees we use the natural terminology. The elements of $T$ are sometimes called \emph{nodes}. The element $\epsilon$ is called the \emph{root} of $\str{T}$, 
nodes $\alpha \in T$ for which there is no $i \in \N$ such that $\langle \alpha, i\rangle \in T$, are called \emph{leaves}.
For a node $\alpha$,
the nodes $\langle \alpha, i \rangle$ are called its \emph{children}, 
the node $\langle \alpha, 0\rangle$ is its  \emph{leftmost child}, 
the node $\langle \alpha, i\rangle$ for which $\langle \alpha, i+1\rangle \not\in T$ 
is its \emph{rightmost child}, 
the node $\beta$ such that $\alpha=\langle \beta, i \rangle$ is its \emph{parent},
the nodes $\langle \alpha, \beta \rangle$ where $\beta \not=\epsilon$ are its \emph{descendants}, 
the nodes $\beta$ such that $\beta$ is a proper prefix of $\alpha$  are its \emph{ancestors},
the node $\langle \alpha, i-1\rangle$ (if $i>0$) is its \emph{previous sibling},
the node $\langle \alpha, i+1\rangle$ (if it belongs to $T$) is  its \emph{next sibling}, 
the nodes $\langle \alpha, j\rangle$ for $j < i$ are its  \emph{preceding siblings} 
and the nodes $\langle \alpha, j\rangle$ for $j > i$ are its \emph{following siblings}. 

The relations $\succv, \lessv, \succh, \lessh$ are called, respectively, the \emph{child}-, \emph{descendant}-, \emph{next sibling}- and \emph{following sibling} relations.
If a tree interprets at least one of $\succh$, $\lessh$, then it is called
an \emph{ordered tree}; otherwise it is an \emph{unordered tree}. 
Trees interpreting all four navigational relations are called {XML} trees.  
Trees in this paper are \emph{unranked}, that is, there is no \emph{a priori} bound on the
number of the children of a node.  

We say that a chain of nodes $\epsilon, \langle i_1\rangle, \langle i_1, i_2\rangle, \ldots, \langle i_1, i_2, \ldots, i_l\rangle$, where the last element
is a leaf, is a  \emph{vertical path}, and a chain of elements $\langle \alpha, 0\rangle, \langle \alpha, 1\rangle, \ldots, \langle \alpha, l\rangle$, where the last element is a rightmost child, is a \emph{horizontal path}. We may speak about vertical (horizontal) paths even if the structure does not interpret $\succv$ ($\succh$).

\subsection{Logics}
Over such structures we consider the \emph{one-dimensional fragment}, \ODF{}, and compare it with several other fragments of first-order logic. 
\ODF{}  is the relational fragment in which
quantification is restricted to blocks of existential quantifiers that leave at most one variable free.
Formally, the set of formulas of \ODF{} over the relational signature $\sigma$ and some countably infinite set of variables $Var$ 
is the smallest set such that:
\begin{itemize}
\item $R\bar{x}$ $\in$ \ODF{} for all $R \in \sigma$ and all tuples $\bar{x}$ of variables from $Var$ of the appropriate length,
\item $x=y$ $\in$ \ODF{} for all variables $x,y \in Var$,
\item \ODF{} is closed under $\vee$ and $\neg$,
\item if $\phi$ is an \ODF{} formula with the free variables $x_0, \ldots, x_k$, then the
formulas $\exists x_0, \ldots, x_k \phi$ and $\exists x_1, \ldots, x_k \phi$ belong to \ODF{}.
\end{itemize}
As usual, we can use standard abbreviations for other Boolean operations, like $\wedge, \rightarrow, \top$, etc., as well as for universal quantification.
The length of a formula $\phi$ is measured as the total number of symbols required to write down $\phi$, and denoted $\sizeOf{\phi}$. The \emph{width} of a formula is the maximum of the numbers of free
variables in its subformulas. 

We will write \ODF$[\sigma_{nav}]$ to indicated that we are interested in \ODF{} formulas over the signature $\sigma_0 \cup \sigma_{nav}$ for
some set $\sigma_0$ of unary relation symbols. We will use the same convention for other logics also.

Some results in this paper will refer to the  \emph{unary negation fragment}, \UNFO{} \cite{SC13}.
The set of \UNFO{} formulas is the smallest set of formulas such that:
\begin{itemize}
	\item $R\bar{x}$ $\in$ \UNFO{} for all $R \in \sigma$ and all tuples $\bar{x}$ of variables from $Var$ of the appropriate length,
  \item $x=y$ $\in$ \UNFO{} for all variables $x,y \in Var$,
   \item \UNFO{} is closed under $\vee$, $\wedge$ and existential quantification,   
	\item if $\phi(x)$ is an  \UNFO{} formula with no free variables besides (at most) $x$ then $\neg \phi(x)$ is also in \UNFO{}.
\end{itemize}

We emphasise that \UNFO{} is not closed under negation, and does not allow for a direct universal quantification.

The following lemma, showing that \UNFO{} may be seen as a fragment of \ODF{} is implicit in \cite{SC13}:
\begin{lem} \label{l:unfotoodf}
There is a polynomial procedure which, given an \UNFO{} formula $\phi$, produces an equivalent formula $\phi'$ in \UNFO{} $\cap$ \ODF{} over the same signature.
\end{lem}
\begin{proof}
In \cite{SC13}, it is shown that any \UNFO{} formula can be converted into the so-called UN-normal form, which is one-dimensional by definition. 
\end{proof}

We note that, generally, no translation from \ODF{} to \UNFO{} exists.  This non-existence is shown in \cite{HK14} for the
extension GNFO of \UNFO. 
Actually, the satisfiability problem (over the class of all structures)
for \UNFO{} is decidable \cite{SC13}, and for \ODF{} it is undecidable \cite{HK14}.

Other relevant fragments of first-order logic which will be mentioned in this paper are the \emph{two-variable} fragment, \FOt{}, 
the \emph{two-variable fragment with counting quantifiers}, \Ct, the two variable version of the
\emph{guarded fragment}, \GFt, 
the navigational core of XPath, \CoreXPath, and the \emph{unary temporal logic} \UTL. 

The formulas of \FOt{} are just those first-order relational formulas which use only the two variables $x$ and $y$.
\GFt{} is the fragment of \FOt{} in which every quantifier is appropriately relativised by an atomic formula
(see, \eg, \cite{Gra99}). \Ct{} extends \FOt{} by counting quantifiers, that is, it adds to \FOt{} constructs of the form
$\exists^{\ge C} y \psi(x,y)$ and $\exists^{\le C} y \psi(x,y)$, for $C \in \N$, with the natural semantics: for $a \in A$, we have that
$\str{A} \models \exists^{\ge C} y \psi(a,y)$ if there are at least $C$ elements $b \in A$ such that $\str{A} \models \psi(a,b)$. 
Analogously for $\exists^{\le C}$.

\UTL{} will be mentioned in the case of words. It is a temporal logic with four navigational operators: \emph{next state},
\emph{somewhere in the future}, \emph{previous state}, \emph{somewhere in the past}, but without binary operators
\emph{since} and \emph{until} (see \cite{EVW02} for more details). 

A corresponding formalism for trees is \CoreXPath{}. We present it here as a modal logic with four pairs of modalities, each pair corresponding to 
one of the relations from the set $\{\succv, \lessv, \succh, \lessh \}$. Definitions in the literature slightly differ from ours,
but the spirit is the same. Let $\Sigma_0$ be a set of propositional variables,
and let us consider the following eight \emph{modalities}:
$\langle \downarrow \rangle$, $\langle \uparrow \rangle$,
$\langle \downarrow_+ \rangle$,  
 $\langle \uparrow^+ \rangle$, 
$\langle \rightarrow \rangle$, 
$\langle \leftarrow \rangle$,
$\langle \rightarrow^+ \rangle$, $\langle \leftarrow^+ \rangle$.
The set of \CoreXPath{} formulas over $\Sigma_0$ is the least set  such that:
\begin{itemize}
\item any $P$ in $\Sigma_0$ is in \CoreXPath,
\item \CoreXPath{} is closed under Boolean connectives
\item if $\psi$ is in \CoreXPath{} then so is $\langle \cdot \rangle \psi$ for any modality $\langle \cdot \rangle$. 
\end{itemize}

Identifying $\Sigma_0$ with $\sigma_0$ (that is, treating propositional variables of $\Sigma_0$ as unary relation symbols in $\sigma_0$), we can interpret \CoreXPath{} formulas over trees. Given 
a tree $\str{T}$ and its node $a$ we inductively define what it means that a \CoreXPath{} formula $\psi$ holds at $a$, written $\str{T}, a \models \psi$.
For $P \in \Sigma_0$ we have $\str{T}, a \models P$ iff $\str{T} \models P(a)$,
$\str{T}, a \models \langle \downarrow \rangle \psi'$ if there is $b \in T$ such that $\str{T} \models a \succv b$ and
$\str{T}, b \models \psi'$, and analogously for the other modalities, which require $\psi'$ to be satisfied at, respectively,
the parent,  a descendant, an ancestor, the next sibling, the previous sibling, a following sibling, and a preceding sibling. 

Using the so-called \emph{standard translation} we can translate \CoreXPath{} formulas to equivalent first-order formulas with one 
free-variable. By an appropriate reuse of variables this translation fits into \FOt$[\succv, \lessv, \succh, \lessh]$,
and actually even in \GFt$[\succv, \lessv, \succh, \lessh]$ (cf.~\cite{MR04}).
As an example, the formula $\langle \uparrow \rangle (P \wedge \langle \rightarrow^+ \rangle (Q \vee \langle \downarrow_+\rangle R))$ 
can be translated to $\exists y (y \succv x \wedge P(y) \wedge \exists x (y \lessh x \wedge (Q(x) \vee \exists y (x \lessv y \wedge R(y))))$.

We remark that a similar translation exists for \UTL{}  \cite{EVW02}.

\subsection{Comparing expressive powers}
In this paper we will compare the  expressive powers of the logics mentioned in the previous paragraph over words and trees.
We will concentrate on the case of formulas with one free variable. This is a natural choice when taking into account the character
of the logics considered: \eg, (the standard translations of) formulas in \CoreXPath{} and \UTL{} always have exactly one free variable and quantified subformulas
in \ODF{}, \GFt{}, \FOt{} and \Ct{} leave at most one variable free. 

Let $\cC$ be a class of structures. We say that a logic ${\sf L_1}$ is \emph{less or equally expressive} than a logic ${\sf L_2}$ over $\cC$, written ${\sf L_1} \preceq {\sf L_2}$ ($\cC$ will
always be clear from the context) if for any 
formula with one free variable $\phi_1(x)$ in ${\sf L_1}$, there is a formula with one free variable $\phi_2(x)$ in ${\sf L_2}$ over the same
alphabet such
that for any structure $\str{A}$ and $a \in A$, we have $\str{A} \models \phi_1(a)$ iff $\str{A} \models \phi_2(a)$.  

If $L_1 \preceq L_2$ and $L_2 \preceq L_1$, then we say that the logics are \emph{equiexpressive} and write $L_1 \equiv L_2$. 
If $L_1 \preceq L_2$ but it is not the case that $L_2 \preceq L_1$, then we say that ${\sf L_1}$ is \emph{(strictly) less expressive} than ${\sf L_2}$ and write $L_1 \prec L_2$.

\subsection{Normal form for \ODF{}}
For the parts of this paper concerning satisfiability, we
introduce a convenient normal form, inspired by the Scott normal form for \FOt{} \cite{Sco62} (a similar normal form is used also in \cite{KK14} for 
the uniform \ODF{} over arbitrary structures).
We say that an \ODF$[\sigma_{nav}]$ formula $\varphi$ is in \emph{normal form} if $\varphi$ has the following shape:
\begin{align} \label{eq:normal}
&\bigwedge_{1\le i \le \mse} \forall y_0 \exists y_1 \ldots y_{k_i} \phie_i
\wedge \bigwedge_{1\le i \le \msu} \forall x_1 \ldots x_{l_i} \phiu_i,
\end{align}
where  $\phie_i=\phie_i(y_0, y_1, \ldots, y_{k_i})$ and
$\phiu_i=\phiu_i(x_1, \ldots, x_{l_i})$ are  quantifier-free.
Note that the width of $\phi$ is the maximum of the set $\{k_{i}+1\}_{1\le i \le \mse} \cup \{l_j\}_{1\le j \le \msu}$.
The following fact can be proved in
the standard fashion.

\begin{lem} \label{l:normalform}
For every \ODW{} formula $\varphi$, one can compute in polynomial time an \ODW{} 
formula $\varphi'$ in  normal form (over the signature extended by some
fresh unary symbols) such that: (i) any model of $\varphi$ can be expanded to a model of $\varphi'$
by appropriately interpreting new unary symbols; (ii) any model of $\varphi'$ restricted to the signature of $\varphi$ is a model of $\varphi$.
\end{lem}
\begin{proof}
(Sketch) 
We successively replace innermost subformulas $\psi$ of $\phi$ of the form $$\exists y_1, \ldots, y_k \phi(y_0, y_1, \ldots, y_k)$$ by
atoms $P_\psi(y_0)$, where $P_\psi$ is a fresh unary symbol, and axiomatize $P_\psi$ using normal form conjuncts:
$\forall y_0 \exists y_1, \ldots, y_k (P_\psi(y_0) \rightarrow \phi(y_0, y_1, \ldots, y_k))$ and
$\forall y_0, y_1, \ldots, y_k$ $(\phi(y_0, y_1, \ldots, y_k) \rightarrow  P_\psi(y_0))$.
\end{proof}

Lemma \ref{l:normalform} allows us, when dealing with satisfiability or when analysing the size and shape of models, to restrict attention to normal form formulas.

\subsection{Types} In this subsection we define the classical  notion of (atomic or quantifier-free) type.
For $k \in \N \setminus \{ 0 \}$ a $k$-\emph{type}  $\pi$ over a signature $\sigma = \sigma_0 \cup \sigma_{nav}$
 is a maximal consistent 
 set of $\sigma$-literals over variables $x_1, \ldots, x_k$ (often indentified with the conjunction of its elements).
This means that $\pi$ is a $k$-type iff:
\begin{itemize}
	\item for each $P \in \sigma_0$ and $1 \le i \le k$ either $Px_i$ or $\neg Px_i$ belongs to $\pi$;
	\item for each $\rightleftharpoons \in \sigma_{nav}$ and $1 \le i,j, \le k$, $i \not=j$, either $x_i \rightleftharpoons x_j$ or $\neg x_i \rightleftharpoons x_j$ belongs to $\pi$;
	\item for each $1 \le i <  j \le k$, either $x_i=x_j$ or  $x_i \not= x_j$ belongs to $\pi$;
	\item if $\sigma_{nav}=\{ \succh, \lessh\}$ (respectively, $\sigma_{nav}$ contains at at least one of $\succv$, $\lessv$), then $\pi$ is satisfiable in a word (resp., tree), \ie, there exists a word (resp., tree) $\str{M}$ and its elements $a_1, \ldots, a_k$ such that $\str{M} \models \pi(a_1, \ldots, a_k)$. 
\end{itemize}  

The last condition can be replaced by a purely syntactic one, listing conditions ensuring consistency with
	a linear or, respectively, tree shape of structures. Listing such conditions would be routine but slightly awkward, so
	we omit them here.
	
A \emph{type} is a $k$-type for some $k \ge 1$. Note that a $1$-type is fully characterized by a subset of $\sigma_0$. 

We say that a tuple of elements $a_1 \ldots, a_k$ of a structure (word or tree) $\str{A}$ \emph{realizes} a $k$-type $\pi$ if $\str{A} \models \pi(a_1, \ldots, a_k)$. In this case we write $\type{A}{a_1, \ldots, a_k} = \pi$. Note that every tuple of elements of a structure realizes precisely one type.

\section{Expressivity of one-dimensional fragment over words} \label{s:expwords}

It is known that the two-variable fragment, \FOt{}, is expressively equivalent over words and $\omega$-words to \UTL{}
\cite{EVW02}. 
It is also
equivalent to \Ct{}, \cite{BCK17}.
Also \GFt{}, as a fragment of \FOt{} containing \UTL{}, has the same expressive power. 
Here we show that \ODF{} and \UNFO{} share this expressivity. To properly handle \UTL{} in the following theorem we  identify 
its formulas with their standard translations to \FOt{} which is a formula with one free variable.

\begin{thm} \label{t:expwords}
Over the class of words and $\omega$-words we have: \UTL $\equiv$ \GFt{} $\equiv$ \FOt $\equiv$ \Ct $\equiv$ \UNFO{} $\equiv$ \ODF{}.
\end{thm}  

Let us first make a simple observation about the equivalence of \UNFO{} and \ODF{}. By Lemma \ref{l:unfotoodf}, \UNFO{} is not more
expressive  than $\ODF$. In the opposite
direction, given any \ODW{} formula we can, using basic logical lows, convert it into a form in which the only non-unary negated
formulas are atomic, \ie, are of the form  $\neg x \succh y$ or $\neg x \lessh y$. They can be quite easily translated
into formulas not using negations at all. Indeed, the former can be expressed as $y \lessh x \vee x=y \vee \exists z (x \succh z \wedge
z \lessh y)$ and the latter as $y \lessh x \vee x=y$. This gives a polynomial translation from \ODW{} into \UNFO$[\succh, \lessh]$.

To complete the proof of Thm.~\ref{t:expwords} we need to show the equivalence of \FOt{} and \ODF{}.
Obviously, \FOt{} is a fragment of \ODF{}. 
It remains to show how to translate \ODF{} into \FOt{}.
The crux is to show how to handle formulas
starting with a block of quantifiers.

\begin{lem}
For any \ODF$[\succh, \lessh]$ formula $\psi = \exists y_1 \ldots, y_k \psi_0(y_0,y_1, \ldots, y_k)$ with the free variable $y_0$ there exists
an \FOt$[\succh, \lessh]$ formula $\psi'$ with one free variable such that for every word or $\omega$-word $\str{M}$ and every $a \in M$,
we have $\str{M} \models \psi[a]$ iff $\str{M} \models \psi'[a]$.
\end{lem}

\begin{proof}
We prove this lemma by induction over the quantifier depth of $\psi$, measured as the maximal nesting depth of maximal blocks of quantifiers rather than of individual quantifiers. W.l.o.g.~we assume that every subformula of $\psi$ starting with such a block indeed has a free variable (if it would not be the
case we could always add a dummy variable).
Let us take any
\begin{equation}
\psi = \exists y_1 \ldots, y_k \psi_0(y_0,y_1, \ldots, y_k),
\end{equation}
and assume that its every subformula starting with a maximal block of quantifiers has an equivalent \FOt{}-formula.
Convert $\psi_0$ into disjunctive form (treating subformulas starting with a quantifier as atoms) and distribute existential quantifiers over disjunctions, obtaining
\begin{equation}\label{dnf}
\psi \equiv \bigvee_{i=1}^{s} \exists y_1 \ldots, y_k \psi_i(y_0,y_1, \ldots, y_k),
\end{equation}
for some $s \in \N$, 
where each $\psi_i$ is a conjunction of literals, subformulas with one free variable of the form 
$\exists z_1, \ldots, z_l \chi(y_j,z_1, \ldots, z_l)$,
and negations of such subformulas. 

Recall that the possible atoms are  $P(y_i)$ for a unary symbol $P$, $y_i \succh  y_j$, $y_i \lessh  y_j$ and $y_i = y_j$, for some $i,j$.

An \emph{ordering scheme} over variables $y_0, \ldots, y_k$ is a formula of the form
$\eta_0(y_{i_0}, y_{i_1}) \wedge  \eta_1(y_{i_1}, y_{i_2}) \wedge \ldots \wedge \eta_{k-1}(y_{i_{k-1}}, y_{i_k})$,
where $\eta_i(v,w)$ is one of the following formulas: $v=w$, $v \succh w$ or $v \lessh w \wedge \neg v \succh w$, and $i_0,i_1, i_2, \ldots, i_k$ is
a permutation of $0, 1,  \ldots, k$. 

Consider now a single disjunct 
$\exists y_1 \ldots, y_k \psi_i(y_0,y_1, \ldots, y_k)$
of (\ref{dnf}) 
and replace it by the following disjunction over all possible ordering schemes $\delta$ over $y_0, \ldots, y_k$:
\begin{eqnarray} \label{ldf}
\bigvee_{\delta} \exists y_1 \ldots, y_k ( \delta (y_0, \ldots, y_k) \wedge \psi_i^\delta(y_0,y_1, \ldots, y_k)),
\end{eqnarray}
where $\psi^{\delta}_i$ is obtained from $\psi_i$ by replacing all atoms $y_i \succh  y_j$, $y_i \lessh y_j$ and $y_i = y_j$, which are not bounded
by the quantifiers from $\psi_i$ by $\top$ or $\bot$,
in accordance with the information recorded in $\delta$. Let us write 
 $\psi_i^{\delta}$ as 
$\bigwedge_{j=0}^k \psi_{i,j}^{\delta}(y_j)$, where
$\psi_{i,j}^{\delta}(y_j)$ consists of the conjuncts with the free variable $y_j$.
We now explain how to translate a single disjunct 
\begin{eqnarray} \label{ldfsd}
\exists y_1 \ldots, y_k ( \delta (y_0, \ldots, y_k)  \wedge \bigwedge_{j=0}^k \psi_{i,j}^{\delta}(y_j))
\end{eqnarray}
of (\ref{ldf}). Let $i_0, i_1,  \ldots, i_k$ be the permutation used to generate $\delta$, and let $t$ be the index such that $i_t=0$. 
By the inductive assumption 
we can replace in each $\psi_{i,j}^{\delta}(y_{j})$  
any conjunct of the form $\exists z_1, \ldots, z_l \chi(y_{j}, z_1, \ldots, z_l)$
by an equivalent two-variable conjunct with one free variable $y_0$. Thus, in turn,
$\psi_{i,j}^{\delta}(y_{j})$ can be replaced by an equivalent \FOt{} formula  $\psi'^{\delta}_{i,j}$ with one free variable.

We finally replace (\ref{ldfsd}) by the conjunction of:
\begin{align} 
\label{e:c1} & \psi'^{\delta}_{i,i_t}(y_0),\\ 
\label{e:c2} &\exists y (\eta_{t-1}(y,y_0) \wedge \psi'^{\delta}_{i,i_{t-1}}(y) \wedge \exists y_0 (  \eta_{t-2}(y_0,y) \wedge \psi'^{\delta}_{i,i_{t-2}}(y_0) \wedge             \ldots  )),\\
\label{e:c3} &\exists y (\eta_{t}(y_0,y) \wedge \psi'^{\delta}_{i,i_{t}}(y) \wedge \exists y_0 (  \eta_{t+1}(y,y_0) \wedge \psi'^{\delta}_{i,i_{t+1}}(y_0) \wedge             \ldots  )),
\end{align}
in which (\ref{e:c1}) takes care of subformulas with the free variable $y_0$, (\ref{e:c2}) takes care of witnesses smaller than (or equal) to $y_0$, passing the word from $y_0$ to the left,
and (\ref{e:c3}) takes care of witnesses greater than (or equal to)  $y_0$, passing the word from $y_0$ to the right. Of course, in all the above formulas we appropriately rename the variables
if necessary, so that only $y_0$ and $y$ are used.
\end{proof}

Having translated formulas starting with blocks of quantifiers, we can easily translate other formulas with one free variable, since they are just
Boolean combinations of the former and unary literals, all of them with the same free variable. 
This gives a translation from \ODW{} to \FOt{}$[\succh, \lessh]$.

Observe that starting from an \ODW{} formula this translation may produce a formula in \FOt$[\succh, \lessh]$ which is exponentially 
longer. Essentially, there are two sources of this exponential blow-up. The first is the transformation
to disjunctive form, and the second is considering all possible permutations of variables quantified in a single block of quantifiers. The question whether this blow-up is necessary is left open.

\section{Satisfiability of one-dimensional fragment over words} \label{s:satwords}

We next turn our attention to satisfiability of \ODF{} over words.
Some upper bounds for the problem can be obtained using the translation to
\FOt{} given in the previous section. As this translation involves an exponential blow-up  and
the satisfiability problem for \FOt{} over words is \NExpTime-complete, this
gives a \TwoNExpTime-upper bound. This could be improved by translating \ODF{}
directly to \UTL{}, which  can be done without problems using the same method. As satisfiability of
\UTL{} is \PSpace-complete, we would get an \ExpSpace-upper bound this way.

However, we can do even better.
We prove that the satisfiability problem for \ODW{} both over words and $\omega$-words is \NExpTime-complete.

To this end we develop  a contraction method
involving a careful analysis of certain similarities between elements in a model and explain how to use
it in order to obtain small model properties for \ODW{} both over words and $\omega$-words.
The complexity result will then easily follow.

\subsection{Profiles}

Now we define \emph{profiles}.
Profiles are intended to abstract the information about relations of a given element to the other elements of a word. Namely, they say what the types
of tuples are (of some bounded size) containing the given element. For convenience we will additionally distinguish types of tuples
built of the elements located to the \emph{left} and to the \emph{right} from the given element.

We say that an element $a$ of a word $\str{M}$ \emph{realizes} (or \emph{has}) a $k$-\emph{profile} $\prof{{M}}{k}{a}=(\cF, \cL, \cR)$ if
$\cF$ is the set of all $s$-types, $1 \le s \le k$, realized by tuples $a_1, a_2, \ldots, a_s$ such
that $a=a_1$,
$\cL$ is the set of all $s$-types, $1 \le s \le k$, realized by tuples $a_1, a_2, \ldots, a_s$ such
that $a=a_1$ and for all $2 \le i \le s$ we have  $a_i \lessh a$; and, analogously, $\cR$ 
is the set of all $s$-types, $1 \le s \le k$, realized by tuples $a_1, a_2, \ldots, a_s$ such
that $a=a_1$ and for all $2 \le i \le s$ we have  $a \lessh a_i$.
Given a profile $\theta$ we will sometimes refer to its components with $\theta.\cF$, $\theta.\cL$ and $\theta.\cR$. Note that
$\theta.\cL \cup \theta.\cR \subseteq \theta.\cF$. Note also that $\theta.\cF$ is determined by $\theta.\cL$ and $\theta.\cR$, and vice versa.

\begin{lem} \label{l:profiles}
Let $\str{M}$ be a word or an $\omega$-word over $\sigma=\sigma_0\cup\{ \succh, \lessh \}$ and $k > 0$ a natural number. Then the number of $k$-profiles realized 
in  $\str{M}$ is bounded by a fixed function $\fh$ exponential in $|\sigma_0|$ and $k$.
\end{lem}

\begin{proof}
We introduce a binary relation $\sim_k$ on $M$ as follows. For $a, b \in M$ we set $a \sim_k b$ iff 
the one-type of $a+i$ is equal to the $1$-type of $b+i$ (or both $a+i$ and $b+i$ do not exist)
for all $-k <  i <  k$. Clearly, $\sim_k$ is an equivalence relation and the
number of its equivalence classes is bounded by $(2^{|\sigma_0|})^{2k-1} + 2k -2 = 2^{|\sigma_0| \cdot (2k-1)} +2k -2$
(the number of combinations of $1$-types of elements $a-k+1, \ldots, a+k-1$ plus the
classes of the first $k-1$ and, in the case of a finite word, the last $k-1$ elements).

We show that if $a \sim_k b$ and $\str{M} \models a \lessh b$ then for every type $\pi$ if $\pi \in \prof{\str{M}}{k}{b}.\cR$
then $\pi \in \prof{\str{M}}{k}{a}.\cR$. Take any $\pi \in \prof{\str{M}}{k}{b}.\cR$ and let $b_1, b_2, \ldots, b_k$, with $b_1=b$ be
its realization. Let $u_1, \ldots, u_k$ be a permutation of $\{1, \ldots, k \}$ such that $u_1=1$ and $\str{M} \models b_{u_i} \lessh b_{u_{i+1}} \vee b_{u_{i}}=b_{u_{i+1}}$ for $1 \le i <  k$, that is a permutation ''sorting'' the elements of the given tuple. 
Let $l$ be the maximal index
such that $\str{M} \models b_{u_i} \succh b_{u_{i+1}} \vee b_{u_{i}}=b_{u_{i+1}}$ for all $1 \le i \le l$. 
Consider now the tuple $a_{u_1}, \ldots, a_{u_k}$, such that $a_{u_1}=a$, $a_{u_i} = a+(i-1)$ for $1 <  i \le l$,
and $a_{u_i}=b_{u_i}$ for $l <  i \le k$. Note that $\type{\str{M}}{a_{u_1}, \ldots, a_{u_k}} = 
\type{\str{M}}{b_{u_1}, \ldots, b_{u_k}}$ and thus also $\type{\str{M}}{a_1, \ldots, a_k} = 
\type{\str{M}}{b_1, \ldots, b_k}=\pi$. Since $a=a_1$ this means $\pi \in \prof{\str{M}}{k}{a}.\cR$.

Strictly analogously we can show that
 if $\pi \in \prof{\str{M}}{k}{a}.\cL$
then $\pi \in \prof{\str{M}}{k}{b}.\cL$. 

Thus, when moving along the elements of a single equivalence class of $\sim_k$ in $\str{M}$ from left to right, the $\cR$-components of the profiles of elements either stay unchanged or diminish, and the $\cL$-components either stay unchanged or grow. 
As the set of types contained in each component is determined by the set of $k$-types in this component, and
as the number of $k$-types in a component can be roughly estimated by $(2^{|\sigma_0|})^k \cdot 5k(k-1)$ 
(the number of possible
assignments of one-types to the elements of a tuple of $k$ elements, times
the number of possible binary connections: equal, $y$ a successor of $x$, $x$ a successor of $y$, $y$ to the left from $x$ but not the successor, $y$ to the right from $x$ but not the predecessor, for every pair of elements)
it follows 
that the $\sim_k$-equivalent elements may have at most 
$2 \cdot (2^{|\sigma_0|})^k \cdot 5k(k-1) +1$
different $k$-profiles (recall that the $\cF$-components are determined by $\cL$- and $\cR$-components). Finally, the total number of $k$-profiles is bounded by 
$(2^{|\sigma_0| \cdot (2k-1)} +2k -2) \cdot (2 \cdot 2^{|\sigma_0| \cdot k} \cdot 5k(k-1) +1)$,
which is indeed exponential in both $k$ and $|\sigma_0|$.
\end{proof}

The notion of profiles can be easily connected to satisfaction of normal
form formulas.
Given a normal form formula $\phi$ of width $k$ we say that a $k$-profile $\theta$
is \emph{compatible} with $\phi$ if 
\begin{itemize}
\item for every conjunct
$\forall x_1 \ldots x_{l_i} \phiu_i(x_1 \ldots x_{l_i})$ of $\phi$ and every $l_i$-type $\pi \in \cF$, we
have $\pi \models \phiu_i$.
\item for every conjunct $\forall y_0 \exists y_1 \ldots y_{k_i} \phie_i(y_0, y_1 \ldots y_{k_i})$ of $\phi$
there is a $(k_i+1)$-type  $\pi \in \theta.\cF$ such that
$\pi \models \phie_i(x_1, \ldots, x_{k_i+1})$.
\end{itemize}

It is straightforward to see:
\begin{lem} \label{l:compatible}
A normal form formula $\phi$ of width $k$ is satisfied in a word ($\omega$-word)
$\str{M}$ iff every $k$-profile realized in $\str{M}$ is compatible with $\phi$.
\end{lem}

\subsection{Contraction}

We are ready to prove the contraction lemma. Namely, we observe that removing a fragment of a word 
between two realizations of the same profile (including one of them and excluding the other) does not change the profiles of the surviving elements.

\begin{lem} \label{l:contraction}
Let $\str{M}=\str{M}_1c\str{M}_2d\str{M}_3$ be a word or $\omega$-word, and $k>0$ a natural number. Assume that $\prof{{M}}{k}{c}=\prof{{M}}{k}{d}$ and $\str{M}'=\str{M}_1c\str{M}_3$. Then,
for every $b \in M'$, we have $\prof{{M'}}{k}{b}=\prof{{M}}{k}{b}$.
\end{lem}
\begin{proof}

Consider the case where $b \in M_1 \cup \{ c \}$. Note that the prefix of $\str{M}$ ending in $b$ is then equal to 
the prefix of $\str{M'}$ ending in $b$. It follows that 
 $\prof{M'}{k}{b}.\cL = \prof{M}{k}{b}.\cL$.
It remains to show that $\prof{M'}{k}{b}.\cR$=$\prof{M}{k}{b}.\cR$. 

To show that $\prof{M'}{k}{b}.\cR$ $\subseteq $ $\prof{M}{k}{b}.\cR$,
take any $s$-type $\pi$, $ 1 \le s \le k$, belonging
to $\prof{M'}{k}{b}.\cR$ and let $b_1, \ldots, b_s$ be a realization of $\pi$ in $\str{M}'$, with $b_1=b$.
Let $u_1, \ldots, u_s$ be a permutation of $\{1, \ldots, s \}$ such that $u_1=1$ and $\str{M}' \models b_{u_i} \lessh b_{u_{i+1}} \vee b_{u_{i}}=b_{u_{i+1}}$ for $1 \le i <  s$. Let $l$ be the maximal index such that $b_{u_l} \in M_1 \cup \{ c \}$. Since $b_{u_1}=b \in M_1 \cup \{ c \}$, $l$
is well defined.
Let $\pi' = \type{\str{M}}{d, b_{u_{l+1}}, \ldots, b_{u_s}}$ and observe that $\pi' \in \prof{\str{M}}{k}{d}.\cR$.
By assumption $\pi' \in \prof{\str{M}}{k}{c}.\cR$, and thus there is a realization $c, a_{u_{l+1}}, \ldots, a_{u_s}$ of $\pi'$
in $\str{M}$. Set $a_{u_i}:=b_{u_i}$ for $1 \le i \le l$.
It is now not difficult to see that $\type{\str{M}}{a_{u_1}, \ldots, a_{u_s}}
= \type{\str{M}'}{b_{u_1}, \ldots b_{u_s}}$ and thus also 
$\type{\str{M}}{a_1, \ldots, a_s}
= \type{\str{M}'}{b_1, \ldots b_{s}}=\pi$. Since $a_1=b_1=b$ it follows that $\pi \in \prof{\str{M}}{k}{b}$.

To show that $\prof{M}{k}{b}.R_i$ $\subseteq$ $\prof{M'}{k}{b}.R_i$ we take any $s$-type
$\pi$, $ 1 \le s \le k$  belonging to $\prof{M}{k}{b}.R_i$ and let $b_1, \ldots, b_s$ be a realization of $\pi$ in $\str{M}$,
with $b_1=b$. Let $u_1, \ldots, u_s$ be a permutation of $\{1, \ldots, s \}$ such that $u_1=1$ and $\str{M} \models b_{u_i} \lessh b_{u_{i+1}} \vee b_{u_{i}}=b_{u_{i+1}}$ for $1 \le i <  s$. Let $l$ be the maximal index such that $b_{u_l} \in M_1 \cup \{ c\}$. 
Again note that $l$ is well defined.
Let $\pi'=\type{\str{M}}{c, b_{u_{l+1}}, \ldots, b_{u_s}}$ and observe that $\pi' \in \prof{\str{M}}{k}{c}.\cR$.
By assumption $\pi' \in \prof{\str{M}}{k}{d}.\cR$, and thus there is a realization $d, a_{u_{l+1}}, \ldots, a_{u_s}$ of $\pi'$
in $\str{M}$ with the $a_j$ from $M_3$. Let $a_{u_i}:= b_{u_i}$ for $1 \le i \le l$. 
It is not difficult to see that $\type{\str{M}'}{a_{u_1}, \ldots, a_{u_s}}=\type{\str{M}}{b_{u_1}, \ldots, b_{u_s}}$ and 
thus also $\type{\str{M}'}{a_{1}, \ldots, a_{s}}=\type{\str{M}}{b_{1}, \ldots, b_{s}}=\pi$
and since $a_1=b$
it follows that $\pi \in \prof{\str{M}'}{k}{b}.\cR$.

The case when $a' \in M_3$ can be treated symmetrically: this time we get the equality of the $\cR$-components of the profiles for free and to show the equality the $\cL$-components we use the equality of the $\cL$-components of the profiles of $c$ and $d$.
\end{proof}

\subsection{Surgery on \texorpdfstring{$\omega$}{w}-words}

In this subsection we  work over $\omega$-words. Namely, we show how to transform a given $\omega$-word into a periodic one
without introducing any new profiles. 

\begin{lem} \label{l:regular}
Let $\str{M}$ be an $\omega$-word and $k>0$ a natural number. Let $\str{M}_0$ be the shortest prefix of $\str{M}$ such that 
it contains all the elements having the $k$-profiles which are realized finitely many times in $\str{M}$. Note that $\str{M}_0$  has length at least $k-1$. Let $a_*$ be the first element not belonging to $M_0$,
and $\theta_*$ its $k$-profile. Let $\str{M}_1$ be the shortest fragment of $\str{M}$ such that
\begin{itemize}
\item it starts at $a_*$,  
\item contains a realization of every $k$-profile	which is
realized in $\str{M}$ infinitely many times, 
\item ends at an element whose successor $b_*$ has $k$-profile $\theta_*$. 
\end{itemize}
Consider the $\omega$-word $\str{M}' = \str{M}_0 \str{M}_1^{\omega}$, that is the word obtained by concatenating $\str{M}$ and
infinitely many copies of $\str{M}_1$.
We will call its initial fragment $\str{M}_0$ and the subsequent copies of $\str{M}_1$ \emph{blocks}.
Let $f: M' \rightarrow M$ be the function returning
for every $a' \in M'$ the element from $\str{M}$ which $a'$ is a copy of.
Then, for every $a' \in M'$,  $\prof{M'}{k}{a'}=\prof{M}{k}{f(a)}$.
\end{lem}

\begin{proof} Let us start with a simple observation.
\begin{claim} \label{c:neighbours}
For every $-k <  i <  k$ either both $a'+i$ and $f(a')+i$ do not exist or their $1$-types are identical.
\end{claim}
\begin{proof}
The claim is obvious if $a'$ and $a'+i$ belong to the same block, and easily follows from the requirement that $a_*$ and $b_*$ have the same $k$-profiles in the other case (for this observe also that $\str{M}_0$ contains at least $k$ elements, which follows from the fact that the profiles of the first $k$ elements of a word are unique).
\end{proof}

Let $g:M \rightarrow M'$ be the partial function defined on $M_0 \cup M_1$ such that $g(a)=a$ if $a \in M_0$ and $g(a)$ is the counterpart of $a$ in the first
copy of $M_1$.

Take any $a' \in \str{M}'$. 
First, let us consider the $\cL$-components of the profiles.
Take any $\pi \in \prof{M}{k}{f(a')}.\cL$ and let a tuple
 $\bar{a}_\pi$ be its realization. Let us write the elements of $\bar{a}_\pi$, in the increasing order, removing duplicates, as $\bar{a}_\pi^{sort}=a^0_1, \ldots,  a^0_{s_0}, a^1_1, \ldots, a^1_{s_1},$ 
$\ldots, a^l_1, \ldots, a^l_{s_l}=f(a')$, where for each $i$, $a^i_1, \ldots, a^i_{s_i}$ is a maximal sequence of consecutive elements of 
$\str{M}$. Observe, using Claim \ref{c:neighbours}, that the structure on the sequence $g(a^0_1), \ldots,  g(a^0_{s_0}), \ldots, g(a^{l-1}_1), \ldots,$
$g(a^{l-1}_{s_{l-1}}), a'-(s_l-1), \ldots, a'-1, a'$ is isomorphic to the structure on $\bar{a}_\pi^{sort}$. 
It follows that $\pi \in \prof{M'}{k}{a'}.\cL$.

Take $\pi \in \prof{M'}{k}{a'}.\cL$. Let
$\bar{a}_\pi$ be its realization, and let  us write the elements of $\bar{a}_\pi$, in the increasing order, removing duplicates, as $\bar{a}_\pi^{sort}=a^0_1, \ldots,  a^0_{s_0}, a^1_1, \ldots, a^1_{s_1},$ 
$\ldots, a^l_1, \ldots, a^l_{s_l}$$=a'$. Take the maximal $u$ such that $a^u_{s_u} \in M_0$. For all $i \le u$ and all $j$ let
$b^i_j:=g(a^i_j)$. Now, for $i=u+1, \ldots, k$ repeat the following. Consider the sequence $f(a^i_{s_i})-s_i+1, \ldots, f(a^i_{s_i})-1 , f(a^i_{s_i})$. By Claim \ref{c:neighbours} the structure on this sequence is isomorphic to the structure on  $a^i_1, \ldots, a^i_{s_i}$. Let $b^i_{s_i}$ be an
element of $\str{M}$ whose profile is identical to the profile of $f(a^i_{s_i})$, and is located at least $k+1$
positions to the right from $b^{i-1}_{s_{i-1}}$. Such an element exists since the profile of $a^i_{s_i}$ is realized
in $\str{M}$ infinitely many times. For $j=1, \ldots, s_i-1$ take $b^i_j:=b^i_{s_i}-s_i+j$. Note that the structure on the sequence 
$b^0_1, \ldots,  b^0_{s_0}, a^1_1, \ldots, b^1_{s_1},$ 
$\ldots, b^l_1, \ldots, b^l_{s_l}$ is isomorphic to the structure on the sequence $\bar{a}_\pi^{sort}$. Thus $\pi \in \prof{M}{k}{b^l_{s_l}}.\cL$. But 
$\prof{M}{k}{b^l_{s_l}}=\prof{M}{k}{f(a^l_{s_l})}=\prof{M}{k}{f(a')}$. So $\pi \in \prof{M}{k}{f(a'}).\cL$.

The reasoning for the equality of the $\cR$-components is similar but simpler and we omit it here.
\end{proof}

\subsection{Complexity}

Using the tools prepared in the previous subsection, we can now show the following small model properties.

\begin{lem} \label{l:small}
Every  normal form \ODW{} formula $\phi$ satisfiable over a finite word has a model of size bounded exponentially in $\sizeOf{\phi}$.
\end{lem}
\begin{proof}
Due to Lemma \ref{l:normalform}, we can assume that $\phi$ is in normal form. Let $k$ be its width.
We take any finite model $\str{M} \models \phi$ and perform on it the contraction procedure from Lemma \ref{l:contraction}, as many times as possible, \ie, if it still contains a pair of distinct elements with the same $k$-profile. By Lemma \ref{l:profiles}
the number of elements in the resulting model $\str{M}'$ is bounded exponentially in $\sizeOf{\phi}$. By Lemma \ref{l:contraction}, the profiles of the
elements in $\str{M}'$ are retained from $\str{M}$. As $\str{M}\models \phi$, these profiles are compatible with $\phi$. By Lemma \ref{l:compatible}, we get that $\str{M}'$  indeed satisfies $\phi$.
\end{proof}

\begin{lem} \label{l:smallregular}
Every  \ODW{} formula $\phi$ satisfiable over an $\omega$-word has a model $\str{N}_0 \str{N}_1^\omega$ where both $|N_0|$ and $|N_1|$ are bounded
exponentially in $\sizeOf{\phi}$.
\end{lem}

\begin{proof}
Due to Lemma \ref{l:normalform} we can assume that $\phi$ is in normal form.  Let $k$ be its width.
We take an arbitrary $\omega$-model $\str{M} \models \phi$. Let $\str{M}=\str{M}_0 \str{M}_1 \str{M}_2$ where $\str{M}_0$ and $\str{M}_1$ are as in Lemma \ref{l:regular}. Using Lemma \ref{l:contraction} for $\str{M}$, contract its fragments $\str{M}_0$ and $\str{M}_1$ to, resp., $\str{N}_0$ and $\str{N}_1$ so that
every $k$-profile from $\str{M}$ is realized at most once in $\str{N}_0$ and at most once in $\str{N}_1$.
 By Lemma \ref{l:profiles} the number of elements in both $\str{N}_0$ and $\str{N}_1$ are bounded exponentially in $\sizeOf{\phi}$.
Note that $\str{N}_0 \str{N}_1 \str{M}_2 \models \phi$.
By Lemma \ref{l:regular}  the $k$-profiles of elements in $\str{N}_0 \str{N}_1^\omega$ are retained from $\str{N}_0 \str{N}_1 \str{M}_2$ and the latter
are realized in $\str{M}$. By Lemma \ref{l:compatible}  we get that $\str{N}_0 \str{N}_1^\omega$ is indeed
a model of $\phi$.
\end{proof}

Finally, we can state the main complexity result of this section.

\begin{thm} \label{t:comp1}
The satisfiability problems for \ODW{} over words ($\omega$-words) is \NExpTime-complete. 
\end{thm}
\begin{proof}
For a given \ODW{} formula $\phi$, convert it into its normal form $\phi'$. Then 
guess a finite model of $\phi'$ of size bounded exponentially as guaranteed by Lemma \ref{l:small} (exponentially bounded initial and periodic parts of a regular $\omega$-model as guaranteed by Lemma \ref{l:smallregular}) and check that 
all the profiles realized in this model (in the model generated by the guessed parts) indeed  are compatible with $\phi'$. 
In the case of finite words the profiles are computed in an exhaustive way: for every element $a$ of the guessed model $\str{M}$ we consider all possible
tuples $a_2, \ldots, a_s$ of at most $k-1$ elements and add $\type{\str{M}}{a, a_2, \ldots, a_s}$ to the profile.  

In the case of $\omega$-words, note that all the  $k$-profiles realized in the periodic model are realized in the finite model in which the periodic 
part is taken $2k$ times ($k$ times assuming that the length of the periodic part is bigger than $1$). Thus, it suffices to compute the profiles in such 
finite model. 
\end{proof}

We also get the following corollary concerning \UNFO{}.
\begin{cor}
The satisfiability problems for \UNFO{} over words ($\omega$-words) is \NExpTime-complete. 
\end{cor}
\begin{proof}
The upper bound follows from Lemma \ref{l:unfotoodf} and Thm.~\ref{t:comp1}. The lower bound follows from \NExpTime-hardness of \FOt{} with only unary relations
(without any structure). 
\end{proof}

\subsection{Undecidable extensions}

The two variable fragment over words, \FOt$[\succh, \lessh]$ remains decidable when extended in various orthogonal directions. Here we show 
that three such important analogous extensions are undecidable in the case of \ODF.

\subsubsection{Data words}
A \emph{data word} ($\omega$-\emph{data word}) is a word ($\omega$-word) with an additional binary relation $\sim$ which 
is required to be interpreted as an equivalence relation, and which is intended to model the equality of data from a potentially infinite alphabet.
Data words are motivated by their connections to XML. 
\FOt{} over data words becomes at least as hard as reachability in Petri nets \cite{BDM11}. Nevertheless, the satisfiability problem
remains decidable.  We show that \ODW{} over data words is undecidable, even in the  absence of $\lessh$. 

\begin{thm} \label{t:und}
The satisfiability problem for \OD{}$[\succh]$ over finite data words and over $\omega$-data-words is undecidable.
\end{thm}

\begin{proof}
We employ the standard apparatus of tiling systems. 
A \emph{tiling system} is a quintuple ${\cT} = \langle C, c_0,c_1,\mathit{Hor}, \mathit{Ver} \rangle$, where $C$ is a non-empty, finite set of \emph{colours},
$c_{0}, c_{1}$ are elements of $C$, and $\mathit{Hor}$, $\mathit{Ver}$ are binary relations on $C$
called the \emph{horizontal} and \emph{vertical} constraints,
respectively.  We say that $\cT$  \emph{tiles} the $m \times n$ grid if there is a function
function $f: \{0,1,\ldots, m-1 \} \times \{0,1,\ldots, n-1 \} \rightarrow C$ such that $f(0,0) = c_0$, $f(m-1,n-1)=c_1$, for
all $0 \le i <m-1$, $0 \le j \le n-1$ we have $\langle f(i,j), f(i+1,j) \rangle$ is in $\mathit{Hor}$,
and for all $0 \le i <m$, $0 \le j < n-1$ we have $\langle f(i,j), f(i,j+1) \rangle$ is in $\mathit{Ver}$.
It is well know that the problem of checking if for a given  tiling system $\cT$ there are $m,n$ such that $\cT$ tiles
the $m \times n$ grid is undecidable. 
The problem remains undecidable if we require $m$ to be even and $n$ odd. 

To show undecidability of the satisfiability problem for \OD{}$[\succh, \sim]$ over finite words we construct a formula $\Phi_{\cT}$ which is satisfied in
a finite word iff $\cT$ tiles the $m \times n$ grid for some even $m$ and odd $n$.
We begin the construction of $\Phi_{\cT}$ with enforcing that its model is a finite grid-like structure, in which the relation $\succh$ forms a snake-like 
path from its lower-left corner to the upper-right corner, and the equivalence relation connects some elements from neighbouring columns. 
See Fig.~\ref{f:grid}. As mentioned, we assume that the number of columns 
is odd and the number of rows is even. We employ the following unary predicates: $B$, $T$, $E_c$, $E_r$, whose intended purpose is to 
mark elements in the bottom row, top row, even columns, and even rows, respectively.

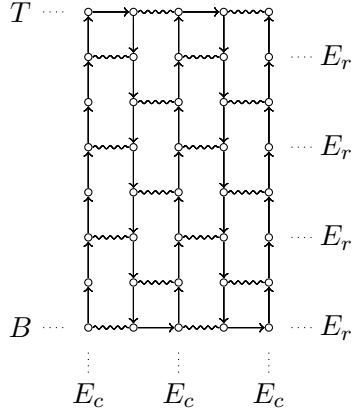
\begin{figure}
\begin{center}
\begin{tikzpicture}[scale=0.6]


\node at (-1.5,0) {$B$};
\draw[dotted] (-1,0) -- (-0.5,0);

\node at (-1.5,7) {$T$};
\draw[dotted] (-1,7) -- (-0.5,7);

\node at (5.5,0) {$E_r$};
\draw[dotted] (4.5,0) -- (5,0);

\node at (5.5,2) {$E_r$};
\draw[dotted] (4.5,2) -- (5,2);

\node at (5.5,4) {$E_r$};
\draw[dotted] (4.5,4) -- (5,4);

\node at (5.5,6) {$E_r$};
\draw[dotted] (4.5,6) -- (5,6);

\node at (0,-1.5) {$E_c$};
\draw[dotted] (0,-1) -- (0,-0.5);

\node at (2,-1.5) {$E_c$};
\draw[dotted] (2,-1) -- (2,-0.5);

\node at (4,-1.5) {$E_c$};
\draw[dotted] (4,-1) -- (4,-0.5);

\foreach \x in {0,1,2,3,4}{
  \foreach \y in {0,1,2,3,4,5,6,7}{
	  \filldraw[fill=white] (\x, \y) circle (0.08);
	
	}

\foreach \x in {0,2,4}
   \foreach \y in {0,1,2,3,4,5,6}{
      \draw[->] (\x, \y+0.1) -- ($(\x,\y) + (0, 0.9)$);
   }

\foreach \x in {1,3}
   \foreach \y in {1,2,3,4,5,6,7}{
      \draw[->] (\x, \y-0.1) -- ($(\x,\y) + (0, -0.9)$);
   }

\draw[->] (0.1,7)--(0.9,7); 
\draw[->] (2.1,7)--(2.9,7); 
\draw[->] (1.1,0)--(1.9,0); 
\draw[->] (3.1,0)--(3.9,0); 

\foreach \x in {0,2}{
  \foreach \y in {0,2,4,6} {
	
	\draw [-,decorate,decoration={snake,amplitude=.2mm,segment length=1mm}] (\x + 0.1, \y) -- (\x + 0.9, \y);
}}

\foreach \x in {1,3}{
  \foreach \y in {1,3,5,7} {
	
	\draw [-,decorate,decoration={snake,amplitude=.2mm,segment length=1mm}] (\x + 0.1, \y) -- (\x + 0.9, \y);
}}

   }
	
\end{tikzpicture}
\end{center}
\caption{The grid-like structure used to show undecidability of \OD{}$[\succh, \sim]$. Solid arrows represent $\succh$, wavy lines represent $\sim$.} \label{f:grid}
\end{figure}

The first two formulas say that the lower left and upper right corners of the grid exist:
\begin{eqnarray}
&\exists x ( Bx \wedge \neg Tx \wedge E_cx \wedge  E_rx \wedge \neg \exists y (y \succh x)) \label{e:first}\\
 &\exists x ( Tx \wedge \neg Bx \wedge E_cx \wedge \neg E_rx \wedge \neg \exists y (x \succh y))
\end{eqnarray}

Next we take care of the $\succh$ relation, ensuring that it respects the intended meaning of the unary predicates:
\begin{eqnarray}
&\forall xy& (x \succh y \rightarrow\\
\nonumber && (E_cx \wedge E_cy \rightarrow( \neg By \wedge \neg Tx \wedge (E_rx \leftrightarrow \neg E_ry)) \wedge\\
\nonumber && (E_cx \wedge \neg E_cy \rightarrow( Tx  \wedge Ty \wedge  \neg Bx \wedge \neg By \wedge \neg E_rx \wedge \neg E_r y)) \wedge\\
\nonumber &&(\neg E_cx \wedge E_cy \rightarrow( Bx  \wedge By \wedge \neg Tx \wedge \neg Ty  \wedge E_rx \wedge  E_r y)) \wedge\\
\nonumber && (\neg E_cx \wedge \neg E_cy \rightarrow( \neg Bx \wedge \neg Ty \wedge (E_r \leftrightarrow \neg E_r y))) )
\end{eqnarray}

Further, we enforce the appropriate $\sim$-connections. (We abbreviate  a formula guaranteeing that 
$x_1, \ldots, x_k$ agree on the $E_c$-predicate by $SameColumn(x_1, \ldots, x_k)$.) 
\begin{eqnarray}
&&\forall xyzt (x \succh y  \wedge y \succh z \wedge z \succh t \wedge Ty \wedge Tz \rightarrow x \sim t)\\
&&\forall xyzt (x \succh y  \wedge y \succh z \wedge z \succh t \wedge By \wedge Bz \rightarrow x \sim t)\\
\nonumber &&\forall xyztuw (SameColumn(x,y,z) \wedge SameColumn(t,u,w) \wedge \\
 && \;\;\;\;\;\;\;\; x \succh y \wedge y \succh z  \wedge z \sim t \wedge t \succh u \wedge u \succh w  \rightarrow x \sim w
\end{eqnarray}

And finally, we say that $T$ and $B$ are appropriately propagated.
\begin{eqnarray}
&&\forall xy (x \sim y \rightarrow (Tx \leftrightarrow Ty) \wedge (Bx \leftrightarrow By))\\
\nonumber &&\forall xyzt (SameColumn(x,y) \wedge SameColumn(z,t) \wedge\\
 && \;\;\;\;\;\;\;\; x \succh y \wedge y \sim z \wedge z \succh t \rightarrow (Tx \leftrightarrow Tt) \wedge (Bx \leftrightarrow Bt))  \label{e:mid}
\end{eqnarray}

Formulas \eqref{e:first}-\eqref{e:mid} ensure that all the vertical segments of the snake-like path are of the same length and thus
that any model indeed looks like in Fig.~\ref{f:grid}.
It remains to encode the tiling problem. We use a unary predicate $P_c$ for each $c \in C$. We say that each node of the grid is coloured by
precisely one colour from $C$, that $(0,0)$ is coloured by $c_0$ and that $(m-1,n-1)$ is coloured with $c_1$:
\begin{eqnarray}
&& \forall x (  \bigvee_{c \in C} P_c(x) \wedge
\bigwedge_{c \neq d} \neg (P_c(x) \wedge P_d(x))),\\
&&\forall x ((\neg \exists y y \succh x) \rightarrow P_{c_0} (x)),\\
&&\forall x ((\neg \exists y x \succh y) \rightarrow P_{c_1} (x)).
\end{eqnarray}
Let us abbreviate by $\Theta_{H}(x,y)$ the formula $\bigwedge_{\langle c, d \rangle \not \in Hor} (\neg P_c(x) \wedge \neg P_d (y))$
stating that $x,y$ respect the horizontal constraints of $\cT$ and by $\Theta_{V}(x,y)$ the analogous formula for vertical constraints.
We  take care of vertical adjacencies:
\begin{eqnarray}
&&\forall xy (E_c(x) \wedge E_c(y) \wedge x \succh y \vee \neg E_c(x) \wedge \neg E_c(y) \wedge y \succh z \rightarrow \Theta_{V}(x,y)),
\end{eqnarray}
and  of horizontal adjacencies:
\begin{eqnarray}
&&\forall xyzt (x \succh y \wedge y \succh z \wedge z \succh t \wedge Ty \wedge Tz \rightarrow \Theta_{H}(y,z)),\\
&&\forall xyzt (x \succh y \wedge y \succh z \wedge z \succh t \wedge By \wedge Bz \rightarrow \Theta_{H}(y,z))),\\
\nonumber &&\forall xyztuw (SameColumn(x,y,z) \wedge SameColumn(t,u,w) \wedge  z \sim t \wedge \\
 && \;\;\;\;\; x \succh y \wedge y \succh z \wedge t \succh u \wedge u \succh w \rightarrow \Theta_{H}(x,w) {\wedge} \Theta_{H}(y,u) {\wedge} \Theta_{H}(z,t)). \label{e:last}
\end{eqnarray}

Let $\Phi_{\cT}$ be the conjunction of 
\eqref{e:first}-\eqref{e:last}. From any model of
$\Phi_{\cT}$, we can read off a tiling of an $m \times n$ grid by
inspecting the colours assigned to the elements of the model. On the
other hand, given any tiling for ${\cT}$, we can construct a finite
model of $\Phi_{\cT}$ in the obvious way. We leave the detailed arguments
to the reader.

The case of $\omega$-words can be treated essentially in the same way. We 
just mark one element in a model, corresponding to the upper-right corner of 
the grid, with a special unary symbol, and relativize all our formulas to positions smaller 
than this element (marked with another fresh unary symbol). In effect, it is irrelevant what happens in the infinite fragment of a model starting from
this marked element.

What is probably worth commenting is that in our undecidability proof we use the equivalence relation $\sim$ in a very
limited way, actually not benefiting from  its transitivity or symmetry. In fact, the transitivity of $\sim$ does not help, being rather 
an obstacle in our construction.
\end{proof}

\subsubsection{Uninterpreted binary relation} \label{s:und}
Both \FOt$[\succh]$ and \FOt$[\lessh]$ remain decidable when, besides $\succh$ or $\lessh$, the signature may contain other binary symbols, whose interpretation
is not fixed (\cite{Ott01}, \cite{CW16b}). We can easily see that this is not the case for \OD{}.

\begin{thm} \label{t:und2}
The satisfiability problem for \ODF{}$[\succh]$ and \ODF{}$[\lessh]$ is undecidable when an additional uninterpreted binary relation is available. 
\end{thm}

\begin{proof}
We can use the proof of Thm~\ref{t:und} without assuming that $\sim$ is an equivalence relation.
\end{proof}

Actually, undecidability holds even without using the linear order: we can simply 
axiomatize grid-like structures using a single binary predicate and some unary coordinate predicates. 
This can be done by a simple modification of the undecidablity proof for \ODF{} over the class of all structure \cite{HK14} which uses two binary symbols.

\subsubsection{Two linear orders}
Let us now consider a variation in which we have two linear orders rather than just one. The second linear order may be interpreted, \eg, as a comparison relation on data values. 
\FOt{}$[\succh_1, \succh_2]$, the two-variable fragment accessing the linear
orders through their successor relations only, is decidable in \NExpTime{} \cite{CW16b}. Showing that a corresponding variant of \ODF{} is undecidable is
again easy. We can define a grid-like structure using the first linear order to form a snake-like path as in the proof of Thm.~\ref{t:und}
and the second to form another snake-like path, starting in the upper-left corner, ending in the lower-left corner and going horizontally 
through our grid, with steps down only on the borders. See Fig.~\ref{f:gridb}. The required structure can be defined with help of some 
additional unary predicates. Since the details of the construction do not differ significantly from the
details of the proof of Thm.~\ref{t:und} we omit them here.

\begin{figure}
\begin{center}
\begin{tikzpicture}[scale=0.6]

\foreach \x in {0,1,2,3,4}{
  \foreach \y in {0,1,2,3,4,5,6,7}{
	  \filldraw[fill=white] (\x, \y) circle (0.08);
	
	}}

\foreach \x in {0,2,4}
   \foreach \y in {0,1,2,3,4,5,6}{
      \draw[->] (\x, \y+0.1) -- ($(\x,\y) + (0, 0.9)$);
   }

\foreach \x in {1,3}
   \foreach \y in {1,2,3,4,5,6,7}{
      \draw[->] (\x, \y-0.1) -- ($(\x,\y) + (0, -0.9)$);
   }

\draw[->] (0.1,7)--(0.9,7); 
\draw[->] (2.1,7)--(2.9,7); 
\draw[->] (1.1,0)--(1.9,0); 
\draw[->] (3.1,0)--(3.9,0);

\foreach \x in {0,1,2,3}
   \foreach \y in {0,2,4,6}{
      \draw[<-, densely dotted] (\x+0.1, \y+0.1) -- ($(\x,\y) + (0.9, 0.1)$);
   }

\foreach \x in {0,1,2,3}
   \foreach \y in {1,3,5,7}{
      \draw[->, densely dotted] (\x+0.1, \y-0.1) -- ($(\x,\y) + (0.9, -0.1)$);
   }

\foreach \x in {0}
   \foreach \y in {6,4,2}{
      \draw[->, densely dotted] (\x+0.1, \y-0.1) -- ($(\x+0.1,\y-0.9)$);
   }

\foreach \x in {4}
   \foreach \y in {7,5,3,1}{
      \draw[->, densely dotted] (\x-0.1, \y-0.1) -- ($(\x-0.1,\y-0.9)$);
   }

\end{tikzpicture}
\end{center}
\caption{The grid-like structure used to show undecidability of \OD{}$[\succh_1, \succh_2]$. Solid arrows represent $\succh_1$, dotted arrows represent $\succh_2$.} \label{f:gridb}
\end{figure}
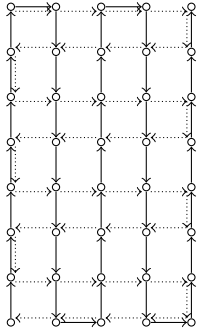

\begin{thm}
The satisfiability problem for \OD{}$[\succh_1, \succh_2]$ is undecidable. 
\end{thm}

\section{Expressivity of one-dimensional fragment over trees} \label{s:exptrees}

In this section we compare the expressive power of \ODF{} with  related logics in the case of two important
tree signatures: $\{ \succv, \lessv, \succh, \lessh \}$ and $\{ \succv, \lessv \}$, or, in other words, over the class
of XML trees and unordered trees. 

As in the case of \UTL{} over words we will identify \CoreXPath{} formulas with their standard translations into \GFt{}, which are
formulas with one free variable.

It turns out that over ordered trees all logics we are interested in are equiexpressive, as it was in the case over words. 

\begin{thm} \label{t:expr1}
For $\sigma=\{\succv, \lessv, \succh, \lessh \}$ we have \CoreXPath{} $\equiv$ \GFt$[\sigma] \equiv$ \FOt$[\sigma] \equiv$  \Ct$[\sigma] \equiv$
 \UNFO$[\sigma] \equiv$  \ODF$[\sigma]$.
\end{thm}

\begin{proof}
Let us first observe that \ODF{} is equivalent to \UNFO{}. The argument is similar to the one in the case of
words. Due to Lemma \ref{l:unfotoodf}, \UNFO{} is not more expressive than \ODF{}. In the opposite direction, 
given any \ODF$[\succv, \lessv, \succh, \lessh]$ formula we can, using basic logical laws, convert it into a form in which the only non-unary negated formulas are atomic, \ie, they are of the form  $\neg x \succv y$, $\neg x \lessv y$, $\neg x \succh y$ or $\neg x \lessh y$. 
Taking into consideration the shape of trees, we translate them into formulas not using negations at all as follows.
\begin{align}
\nonumber trans(\neg x \lessv y)  := &\exists z ((z=x \vee z \lessv x) \wedge (y=z \vee \exists t (z \lessh t \vee t \lessh z) \wedge (y=t \vee t \lessv y))) \\
\nonumber trans(\neg x \succv y) :=& trans( \neg x \lessv y) \vee \exists z (x \succv z \wedge z \lessv y ) \\
\nonumber trans(\neg x \lessh y)  := & x \lessv y \vee \exists z (z \lessh x \wedge (y=z \vee z \lessv y )) \vee
                        \exists z (x \lessh z \wedge z \lessv y) \vee \\
												&\exists z (z \lessv x \wedge (y=z \vee \exists t (z \lessh t \vee t \lessh z) \wedge (y=t \vee t \lessv y)))\\
\nonumber trans(\neg x \succh y)  := & trans(\neg x \lessh y) \vee \exists z (x \succh z \wedge z \lessh y ) 
\end{align}

This gives a polynomial translation from \ODF$[\succv, \lessv, \succh, \lessh]$ into \UNFO$[\succv, \lessv, \succh, \lessh]$.
and establishes their equivalence.

Further, a translation from \UNFO$[\succv, \lessv, \succh, \lessh]$ 
to \CoreXPath{} 
is given in \cite{SC13}. \CoreXPath{} is a fragment of \GFt{}, \GFt{} is a fragment of \FOt{}, and \FOt{} is a fragment of
\Ct. Finally,  
\Ct{} can be easily translated to \ODF{} (over any class of structures): consider, \eg, a subformula of the form $\exists^{\ge k}y \psi(x,y)$
and note that it 
can be written as $\exists y_1, \ldots, y_k (\bigwedge_{i \not= j} y_i \not= y_j \wedge \bigwedge_i \psi(x,y_i))$. 
This completes the proof in the case of XML trees. 
\end{proof}

Over unordered trees the situation if more interesting and the considered languages turn out to vary in their expressive power, in particular \ODF{} is more expressive than \UNFO{} and \FOt{}. Interestingly it is, however, equivalent to \Ct{}. The full picture is as follows.

\begin{thm} \label{t:expr2}
For $\sigma=\{ \succv, \lessv \}$ we have  \CoreXPath$[\sigma]$ $\equiv$ \GFt$[\sigma] \equiv$ \UNFO$[\sigma] \prec$ \FOt$[\sigma]  \prec $ \Ct$[\sigma] \equiv$ \ODF$[\sigma]$.
\end{thm}

\begin{proof}
Let us assume that the signature contains no unary predicates and for $i \in \N$ let $\str{T}_i$ denote the tree consisting just of a
  root and its $i$ children. The \Ct{} formula 
  $\exists^{\ge 3} y \; x \lessv y$ distinguishes $\str{T}_3$ and $\str{T}_2$ (it is true at the root of the former and false in
	the latter), 
	while the \FOt{}
	formula $\exists y (\neg x \lessv y \wedge \neg y \lessv x \wedge x \not= y)$ distinguishes 
	$\str{T}_2$ and $\str{T}_1$. It is not difficult to see that
	\FOt{} cannot distinguish between $\str{T}_3$ and $\str{T}_2$ (a simple
	2-pebble game argument, cf.~\cite{BCK17}) and that \GFt{} cannot
	distinguish between $\str{T}_2$ and $\str{T}_1$ (use guarded bisimulations, cf.~\cite{ABN98}). 
	These observations justify the relations \GFt{} $\prec$ \FOt{} and \FOt{} $\prec$ \Ct{}.

\Ct{} can be translated to \ODF{} as in the previous proof.
Translation in the opposite direction is a harder task and 	we devote for it a separate subsection.

It remains to show the equivalence of \CoreXPath, \GFt{} and \UNFO{}. To this end we provide a 
translation from \GFt{} to \UNFO{} and from \UNFO{} to \CoreXPath{}. The cycle is then closed
by recalling that  \CoreXPath{} is a fragment of \GFt{}.

\medskip \noindent
	\emph{From \GFt{} to \UNFO{}}. Take any \GFt{}$[\succv, \lessv]$ formula and write it
	without using the universal quantifiers. Then push down all the negations with the exception of those 
	standing just before the existential quantifiers (they are allowed in \UNFO{} since a \GFt{} formula
	starting with an existential quantifier has at most one free variable). 
	Let $\phi$ be the resulting formula. We need to eliminate from $\phi$ all occurrences of negated binary literals.
	We will do this in a bottom-up manner.
	
	Take an innermost subformula $\psi$ of $\phi$ starting with a maximal block of quantifiers. If $\psi=\exists x \psi_0(x)$ 
	or $\psi=\exists y \psi_0(y)$ then there is nothing to do, as there are no negated binary literals
	in $\psi$.
	Otherwise $\psi$ has one of the three forms: $\exists y (\alpha(x,y) \wedge \psi_0(x,y))$ or $\exists x (\alpha(x,y) \wedge \psi_0(x,y))$
	or $\exists x, y (\alpha(x,y) \wedge \psi_0(x,y))$, where  $\alpha$ is
	one of the four possible guards $x \succv y$, $y \succv x$, $x \lessv y$, $y \lessv x$. 
	
	Consider the first form (the other one can be treated similarly).
	If $\alpha= x \succv y$ then replace any literal $\neg x \succv y$ by $\bot$, $\neg y \succv x$ by $\top$, 
	$\neg x \lessv y$ by $\bot$ and $\neg y \lessv x$ by $\top$. If $\alpha = y \succv x$ then proceed analogously.
	If $\alpha = x \lessv y$ then first convert $\psi$ into the equivalent formula
	$(\exists y (x \succv y \wedge \psi_0(x,y))) \vee (\exists y (x \lessv y \wedge \neg x \succv y \wedge \psi_0(x,y)))$.
	With the first disjunct we proceed as described above. The second one is replaced by $\exists z, y (x \succv z \wedge z \lessv y \wedge \psi_0(x,y)$
	and then in $\psi_0(x,y)$ we replace $\neg x \succv y$ by $\top$, $\neg y \succv x$ by $\top$, 
	$\neg x \lessv y$ by $\bot$ and $\neg y \lessv x$ by ${\top}$. If $\alpha=y \lessv x$ then we proceed analogously.
	In all cases we obtain an \UNFO{} replacement of $\psi$. We then proceed up in the syntax tree of the input formula and finally 
	end up with a \UNFO{} formula equivalent to $\phi$.
	
	\medskip \noindent
	\emph{From \UNFO{} to \CoreXPath{}}. Let $\phi$ be a formula in \UNFO$[\succv, \lessv]$. Recall that by Lemma \ref{l:unfotoodf} we may assume that
		$\phi \in \UNFO{} \cap \ODF$. Here we proceed similarly as in the translation from $\ODF[\succh, \lessh]$ to $\FOt[\succh, \lessh]$
	in the proof of Thm.~\ref{t:expwords}. Again, the crux is to show how to translate the subformulas of $\phi$ starting with a block of quantifiers.
	Assume that 
	\begin{equation}
	\psi=\exists y_1, \ldots, y_k \psi_0(y_0, y_1, \ldots, y_k) 
	\end{equation}
	is such a subformula. 
	W.l.o.g.~we may additionally assume that every subformula of $\psi$ starting with a maximal block of existential quantifiers has a free variable
	(if it was not the case, we could add a dummy free variable). 
	
	Convert $\psi_0$ into disjunctive form (treating its subformulas starting with a quantifier as atoms) and distribute existential quantifiers over disjunctions, obtaining
\begin{equation}\label{trees:dnf}
\psi \equiv \bigvee_{i=1}^{s} \exists y_1 \ldots, y_k \psi_i(y_0,y_1, \ldots, y_k),
\end{equation}
for some $s \in \N$, 
where each $\psi_i$ is a conjunction of unary literals, binary atoms of the form $y_i \succv y_j$, $y_i \lessv y_j$ or $y_i=y_j$, and
subformulas  of the form 
$\exists z_1, \ldots, z_l \chi(y_j,z_1, \ldots, z_l)$  (with one free variable) or their negations. Note that we do not have negated binary literals.

A  \emph{tree ordering scheme} over variables $y_0, \ldots, y_k$ is a conjunction $\delta$ of atoms of the form 
$y_i \succv y_j$, $y_i \lessv y_j$ or $y_i=y_j$ that can be satisfied in a tree in such a way that
this tree satisfies no  binary atoms over $y_0 \ldots, y_k$ except those from $\delta$. For example $y_0 \succv y_1 \wedge  y_0 \succv y_2 \wedge y_3 \lessv y_4$ is
a tree ordering scheme, but $y_0 \lessv y_1 \wedge  y_0 \lessv y_2 \wedge y_1 \lessv y_3 \wedge y_2 \lessv y_3 \wedge y_0 \lessv y_3$ is not since to satisfy it,
one needs to add (at least) either $y_2 \lessv y_1$ or $y_1 \lessv y_2$.

Consider now a single disjunct 
$\exists y_1 \ldots, y_k \psi_i(y_0,y_1, \ldots, y_k)$
of (\ref{trees:dnf}),
and replace it by the following disjunction over all possible tree ordering schemes $\delta$ over $y_0, \ldots, y_k$ containing 
all the binary atoms of $\psi_i$ which are not bounded
by the quantifiers of the subformulas of $\psi_i$:
\begin{eqnarray} \label{trees:ldf}
\bigvee_{\delta} \exists y_1 \ldots, y_k ( \delta (y_0, \ldots, y_k) \wedge \psi_i^\delta(y_0,y_1, \ldots, y_k)),
\end{eqnarray}
where $\psi^{\delta}_i$ is obtained from $\psi_i$ by removing all the binary atoms (except those bounded
by the quantifiers from $\psi_i$), as they are now present in $\delta$.

 Let us now write the conjunction $\psi_i^{\delta}$ as  $\bigwedge_{j=0}^k (\mu_{i,j}^{\delta}(y_j) \wedge \nu_{i,j}^{\delta}(y_j))$, where
 $\mu_{i,j}^{\delta}(y_j)$ consists of the literals with free variable $y_j$ and $\nu_{i,j}^{\delta}(y_j)$ consists 
of the subformulas starting with a maximal block of quantifiers with free variable $y_j$. 
We now explain how to translate a single disjunct 
\begin{eqnarray} \label{trees:ldfsd}
\exists y_1 \ldots, y_k ( \delta (y_0, \ldots, y_k) \wedge \bigwedge_{j=0}^k (\mu_{i,j}^{\delta}(y_j) \wedge \nu_{i,j}^{\delta}(y_j)))
\end{eqnarray}
of (\ref{trees:ldf}). 

The idea is to start from $y_0$ and "visit the other variables" in the order specified by $\delta$; first go down from $y_0$ then up, and finally
jump to the nodes whose connection to $y_0$ is not required by $\delta$, 
at each visited node making all the necessary forks,  including those required by the $\nu_{i,j}$. 
Since making the above intuition formal is cumbersome let us just illustrate it by a representative example.

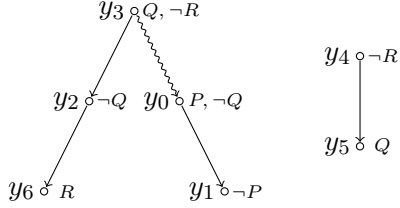
\begin{figure}
\begin{center}
\begin{tikzpicture}[scale=0.6]

\foreach \x/\y in {0/0, 1/2, 2/4, 3/2, 4/0, 7/3, 7/1}{
  	  \filldraw[fill=white] (\x, \y) circle (0.08);
	
	}
	
	\draw[->] (2-0.05,  4-0.1) -- (1+0.05,   2+0.1);

	\draw[->,decorate,decoration={snake,amplitude=.2mm,segment length=1mm}] (2+0.05,  4-0.1) -- (3-0.05,   2+0.1);
	
	\draw [->] (1-0.05,2-0.1) -- (0+0.05,0+0.1);
	
	\draw [->] (3+0.05,2-0.1) -- (4-0.05,0+0.1);
	
	\draw [->] (7,3-0.1) -- (7,1+0.1);
	
	\node at (-0.5,0) {$y_6$};
	\node at (0.5,2) {$y_2$};
  \node at (1.5,4) {$y_3$};
  \node at (2.5,2) {$y_0$};
  \node at (3.5,0) {$y_1$};
  \node at (6.5,3) {$y_4$};
  \node at (6.5,1) {$y_5$};

  \node at (0.5,0) {\tiny $R$};
	\node at (1.5,2) {\tiny $\neg Q$};
  \node at (2.8,4) {\tiny $Q, \neg R$};
  \node at (3.8,2) {\tiny $P, \neg Q$};
  \node at (4.5,0) {\tiny $\neg P$};
  \node at (7.5,3) {\tiny $\neg R$};
  \node at (7.5,1) {\tiny $Q$};

\end{tikzpicture}

\caption{A visualisation of a tree ordering scheme. Straight arrows represent $\succv$, wavy arrows represent $\lessv$. Connections implied by
transitivity are omitted for clarity.}%
\label{f:unfotocore}%
\end{center}
\end{figure}

Let 
$\delta(y_0, \ldots, y_6)=y_3 \succv y_2 \wedge y_3 \lessv y_0 \wedge y_3 \lessv y_6 \wedge y_3 \lessv y_1 \wedge y_0 \succv y_1 \wedge y_2 \succv y_6 \wedge y_4 \succv y_5$.
Let
$\mu_{i,0}=P(y_0) \wedge \neg Q(y_0)$, 
$\mu_{i,1}=\neg P(y_1)$, 
$\mu_{i,2}= \neg Q(y_2)$, 
$\mu_{i,3}= Q(y_3)\wedge \neg R(y_3)$, 
$\mu_{i,4}= \neg R(y_4)$,
$\mu_{i,5}= Q(y_5)$,
$\mu_{i,6}= R(y_6)$, 
and let $\nu_{i,0}=\exists z_1, z_2 \gamma_0(y_0, z_1, z_2)$, $\nu_{i,2}=\exists z_1 \gamma_2(y_2, z_1)$; the other $\nu_{i,j}$ are empty. See Fig.~\ref{f:unfotocore}.
By the inductive assumption we have CoreXPath formulas $\nu'_0$ and $\nu'_2$ equivalent to $\nu_{i,0}$ and, respectively, $\nu_{i,2}$. 
The translation looks then as follows:

\begin{align}
P \wedge \neg Q \wedge \nu'_0   &  \\ 
\nonumber& \wedge <\downarrow> \neg P\\ 
\nonumber& \wedge <\uparrow^+> (Q \wedge \neg R \\
 \nonumber& \;\;\;\;\;\;\;\; \wedge <\downarrow> (\neg Q \wedge \nu_2 \\
\nonumber& \;\;\;\;\;\;\;\; \wedge <\downarrow> R)) \\
\nonumber&\wedge <\uparrow^+><\downarrow_+> (\neg R \\
\nonumber& \;\;\;\;\;\;\;\; \wedge <\downarrow> Q)   
\end{align} 

The first line describes what happens at $y_0$, the second corresponds to a step down to $y_1$, lines 3-5 correspond to a step up to $y_3$
(from which we go down to $y_2$ and then once more down to $y_6$). In lines 6-7 we jump to $y_4$ and then step down to $y_5$.
Note that $y_4$ and $y_5$ do not need to be related to the other variables, which is captured by
$<\uparrow^+><\downarrow_+>$ which works as the universal modality and allows us to move to any part of the tree.

The correctness of the translation relies on the fact that our formulas do not have literals with negated binary atoms,
and thus we do not need to worry about them. Indeed, in our example above one can construct models in which, say, $y_5 \succv y_3$,
or $y_2=y_0$ holds.

This gives a procedure translating  formulas from \UNFO{} $\cap$ \ODF{} starting with a block of existential quantifiers into \CoreXPath.
An arbitrary \UNFO{} $\cap$ \ODF{} formula with one free variable is a Boolean combination of such formulas (with the same free variable),
and its translation is just the Boolean combination of the translations of its constituents. 
	\end{proof}
	
\subsection{Translation from  \ODF$[\succv, \lessv]$ into   \Ct$[\succv, \lessv]$}

In this subsection we provide the missing translation from the proof of Thm.~\ref{t:expr2}.

Before presenting the details, we comment informally on the intuition 
behind the  proofs. The main issue is that \ODF{} formulas
can be presented in a normal form where every subformula $\exists\overline{x}\psi$
essentially describes a substructure with $|\overline{x}|$ (or $|\overline{x}| +1$)
elements. In this normal form, when considering trees only,
subformulae $\exists\overline{x}\psi$ describe substructures of trees. Such
substructures consist essentially of disjoint subtrees (called components).
Such components can be described in \Ct{} by working inductively from the leaves of the
subtrees towards the root. However, with the help of the relation $\lessv$,
we can write our $F_1$ formula in such a form that formulae $\exists\overline{x}\psi$
contain also the unique root of the full tree, and thus the disjoint
components are connected (using $\lessv$). Thus we can write a formula that 
describes the components and also puts them together in a suitable way.
We then turn to the formal argument.
Let $\sigma$ be an arbitrary finite relational vocabulary. (Note indeed that $\sigma$ is
not necessarily a vocabulary for trees.)
%
%
We allow $\sigma$ to contain \emph{nullary} predicates, \ie, Boolean
variables; a nullary predicate $Q\in\sigma$ is an
atomic formula such that any $\sigma$-model $\mathfrak{M}$
interprets $Q$ either such that ${\mathfrak{M}}\models Q$ or
such that ${\mathfrak{M}}\not\models Q$.
A $\sigma$-\emph{diagram} of
\emph{width} $k\in\mathbb{Z}_+$ is a quantifier-free conjunction
consisting of the following.\footnote{We note that diagrams of width $k$ are
similar to $k$-types.}
\begin{enumerate}
\item
A conjunction expressing that the variables $x_1,...,x_k$ are
mutually pairwise distinct.
\item
A conjunction containing \emph{exactly one} of the 
literals $R\overline{x},\neg R\overline{x}$ for
each $R\in\sigma$ and each $\overline{x}\in\{x_1,...\, ,x_k\}^{ar(R)}$,
where $ar(R)$ is the arity of $R$.
\end{enumerate}

\smallskip

A quantifier-free formula is a diagram if it is a $\sigma$-diagram of width $k$
for some $\sigma$ and some $k$.
Let $\exists\overline{x}\varphi$ be a formula of \ODF{}.
The formula $\varphi$ is a Boolean combination of atoms and
existential formulae $\exists \overline{y}\psi$. We call
such existential formulae $\exists \overline{y}\psi$ \emph{relative atoms} of $\varphi$.
The other atoms are called
\emph{free atoms} of $\varphi$. For example, the formula $R(x,y)
\wedge \exists x (S(y,x)\wedge\psi'(y))$
contains a binary free atom $R(x,y)$ and a unary
relative atom $\exists x (S(y,x)\wedge\psi'(y))$.
A formula $\chi$ of \ODF{} is said to be in \emph{diagram normal form} if
every subformula $\exists\overline{x}\varphi$ of $\chi$
has the property that $\varphi$ is a diagram
with respect to the vocabulary $\sigma'$ defined as follows.
\begin{enumerate}
\item
$\sigma'$ contains the predicate symbols of the free atoms of $\varphi$.
\item
If $\psi(x)$ is a unary relative atom of $\varphi$ (i.e., a relative atom with one free variables, $x$),
then this relative atom is considered a unary predicate in $\sigma'$.
\item
Similarly, if $\psi'$ is a nullary relative atom of $\varphi$,
then this relative atom is considered a nullary predicate in  $\sigma'$.
\end{enumerate}
\begin{lem}\label{diagrams}
Formulae of \ODF{} have equivalent
representations in diagram normal form.
\end{lem}
\begin{proof}
Consider an \ODF{}-formula $\exists\overline{x}\varphi$.
Now, the formula $\varphi$ can be modified into a disjunctive normal form
formula $\varphi_{DNF}$ (without modifying its atoms or relative atoms). Then the
quantifier prefix $\exists \overline{x}$ can be distributed
over the disjunctions of $\varphi_{DNF}$.
The resulting formula is of the form $(\exists\overline{x}\varphi_1)
\vee...\vee(\exists\overline{x}\varphi_k)$, where the formulae $\varphi_i$
are conjunctions of (possibly negated) free atoms
and (possibly negated) relative atoms of $\varphi$.
Consider one such disjunct $\exists \overline{x} \varphi_i$.
We first apply the induction hypothesis to
the free atoms and relative atoms of $\varphi_i$ 
and put them in diagram normal form.
Let the obtained formula be denoted by $\varphi_i'$.
We next describe how the formula $\exists \overline{x} \varphi_i'$ can be put to
diagram normal form.
The intuitive idea is roughly to (1) consider all distributions of equalities and 
inequalities for the free variables in $\varphi_i'$, (2) for each such 
distribution, define all
diagrams consistent with $\varphi_i$, and (3) take 
the disjunction over all the diagrams and distribute 
the existential quantifiers of $\exists \overline{x}$ over
the disjunction.

The general formal construction is straightforward but 
cumbersome, so we begin by sketching some examples. First consider
the case where $\varphi_i'$ is $R(x,y)$.
\begin{itemize}
\item
The formula $R(x,y)$ gives
rise to---for example---the diagram $x\not= y\wedge R(x,y)
\wedge\neg R(y,x)\wedge \neg R(x,x)\wedge R(y,y)$. 
\item
Furthermore, the same formula $R(x,y)$ gives
rise to the diagram $R(x,x)$ 
obtained under the distribution of equalities 
which requires that $x=y$. We note that $R(x,x)$ loses the
free variable $y$, which is an undesirable effect. Therefore we
shall use the formula $x=y \wedge R(x,x)$ instead. This in not 
strictly speaking a diagram, but we shall discuss below how to deal with this problem. 
\end{itemize}
To obtain the diagram normal form formula equivalent to $\exists\overline{x}\varphi_i'$ in
this particular case where $\varphi_i'$ is $R(x,y)$, 
we take a disjunction of all the diagrams consistent with $R(x,y)$ and also the formulas that are
not strictly speaking diagrams, such as $x=y \wedge R(x,x)$, and
distribute the existential quantifiers $\exists\overline{x}$ inwards over the disjunctions.
To put the non-diagram conjuncts such as $\chi := \exists \overline{x}(x=y \wedge R(x,x))$ into
diagram normal form, we consider 
three cases. In each case we replace $\chi$ by an equivalent formula in diagram normal form.
Firstly, if $\exists \overline{x} = \exists x$, then we 
replace $\chi$ by $R(y,y)$. Secondly, if $\exists \overline{x} = \exists y$,
we replace $\chi$ by $R(x,x)$. Finally, if $\exists\overline{x} = \exists x \exists y$,
then we replace $\chi$ with $\exists x R(x,x)$. These new formulae are in 
diagram normal form. All remaining cases are similar.
The general way to deal with any
formula $\exists \overline{x}\varphi_i'$ is as follows. Recall that $\varphi_i'$ is a
conjunction of formulas in diagram
normal form. We may assume, w.l.o.g., that $\varphi_i'$ is quantifier-free (because the 
relative atoms will be treated as if they were atoms).
If the quantifier-free formula $\varphi_i'$ is inconsistent, then we simply take an
arbitrary formula $\psi$ in diagram normal form and replace $\exists\overline{x}\varphi_i'$
with $\psi \wedge \neg \psi$. Assuming $\varphi_i'$ is consistent, we perform the following steps.

\begin{enumerate}
\item
Let $X$ denote the set of (free) variables in the conjunction $\varphi_i'$, and 
let $\tau$ be the vocabulary of $\varphi_i'$.
By a \emph{proto diagram over $\varphi_i'$} we mean a
conjunction $\delta \wedge \varphi_i''$ such that the following conditions hold.
\begin{itemize}
\item
$\delta$ is a conjunction of equality atoms and negated equality atoms in
the variables $X$ that contains, for any two distinct variables $x$ and $y$, 
either the conjunct $x=y$ or $x\not= y$. Furthermore, $\delta$ is consistent with the conjunction $\varphi_i'$
and therefore contains 
all equalities and inequalities that already occur in $\varphi_i'$.
\item
By a \emph{relational $\tau$-literal over $X$} we mean a atom or
negated atom (of vocabulary $\tau$ and with its variables from $X$) 
that is not an
equality or negated equality atom.
Now, $\varphi_i''$ is a maximal set of relational $\tau$-literals over $X$
that is consistent with $\varphi_i'$.  
\end{itemize}
We define $\chi'$ to be the disjunction of all 
proto diagrams $\delta \wedge \varphi_i''$ over $\varphi_i'$.

\item
We distribute the existential quantifiers $\exists \overline{x}$ over the disjuncts of $\chi'$.
Consider one such disjunct $\exists \overline{x}(\delta \wedge \varphi_i'')$.
Eliminate positive equalities from $\delta \wedge \varphi_i''$ one-by-one by 
modifying $\delta \wedge \varphi_i''$ such that for each positive equality (e.g., $x=y$), choose
one of the variables (e.g., $x$) and replace each instance of the
chosen variable in $\varphi_i''$ by the other variable (e.g., replace instances of $x$ by $y$). In the 
process, do not
eliminate the possible free variable of $\exists \overline{x}(\delta \wedge \varphi_i'')$ (note that there is at
most one such possible free variable since we are dealing with $\mathrm{F}_1$).
This way we get rid of positive equalities, and the obtained formula is in diagram normal form.
\end{enumerate}

\noindent
This process converts the original formula $\varphi$ into diagram normal form.
\end{proof}
Above we considered general \ODF{}, but now we return to
considering \ODF$[\succv, \lessv]$ over trees
(labelled with extra unary predicates).
A formula of \ODF$[\succv, \lessv]$ is in \emph{rooted diagram form} if
every diagram contains a point $z$ and an extra conjunct $\gamma(z)$ stating
that $z$ is the root of the tree; thus $\gamma(z)$ states that $z$ is an
element that does not have any parent, $\gamma(z) := \neg\exists y(y\succv z)$.
It is easy to modify diagram
normal form formulae into rooted diagram form by 
taking big disjunctions of all possibilities for the 
location of $z$ in a diagram: potentially any 
variable of the original diagram can be the root, and additionally, it may be that none of
the variables is the root and thus a root must be added. We take a 
disjunction of all the possible configurations that arise.
Note that in rooted diagram form, every diagram is connected,
meaning that for all distinct variables $x,y$, there is an
undirected path using the relations $\succv$, $\lessv$ 
that begins from $x$ and ends with $y$.
This connectedness property will help us with the inductive argument in
the proof of the following lemma.
\begin{lem}\label{exprmail}
Let $\varphi(x) = \exists\overline{y}\psi$ be a
formula of \ODF$[\succv, \lessv]$ where $x$ a free variable. Then there
exists a formula $\varphi^*(x)$ of \Ct$[\succv, \lessv]$ that is equivalent to $\varphi(x)$
%
%
%
\end{lem}
\begin{proof}
We first note that if $\varphi(x)$ is not satisfiable over trees, the
desired \Ct{} formula is $\bot$.
Thus we assume that $\varphi(x)$ is satisfiable over trees.
We let $\sigma$ denote the vocabulary of $\varphi(x)$ 
and $\sigma_0$ the proposition symbols that occur in $\varphi$, i.e., 
the unary relation symbols in $\varphi$.
As argued above, we may assume
that $\varphi(x)$ is in rooted diagram normal form.
Therefore $\psi$ is a 
diagram that describes a substructure $M_{\psi}$ of a tree; $M_{\psi}$
contains a root node and is connected.
The nodes of $M_{\psi}$ are the variables of $\psi$.

To express $\varphi(x) = \exists\overline{x}\psi$ in \Ct{}, we
first define, for each leaf $u$ of $M_{\psi}$, the
unique $\sigma_0$-diagram of width $1$ that is satisfied by $u$
and denote this formula by $\delta_u(y)$.
The idea is then to proceed inductively upwards 
from the leaves and define, for each point $v$ in the diagram $M_{\psi}$, a
formula that characterizes---in a way specified later on---the substructure of $M_{\psi}$ below $v$ in the tree.
However, since the free variable $x$ of $\varphi(x)$ is not necessarily the root of $M_{\psi}$, 
and since we want the induction to end at $x$, we will in
fact proceed such that from the perspective of
the \emph{undirected} tree induced by $M_{\psi}$, we
work from the leaves towards $x$. Since $x$ is a 
legitimate root for the underected tree, the induction ends at $x$.
We illustrate the flow of the induction by an example first. Assume
the diagram $M_{\psi}$ consists of the eight points $r,s,s',w,x,x',y$ such that $M_{\psi}$
satisfies the following conditions.
\begin{enumerate}
\item
The the child relation $\succv$ in the
diagram $M_{\psi}$ is $\{(w,x),(w,x'),(x,y),(s,s')\}$.
\item
$r$ connects via $\lessv$ to $s$ and $w$ (and 
thus to $x,x',y,s'$ as well).
\item
If $u\in\{s,s'\}$ and $t\in\{w,x,x',y\}$, then $u\not\downarrow_+ t$ and $t\not\downarrow_+ u$.
\end{enumerate}

\smallskip

In this case we first deal with the leaves $s'$, $x'$ and $y$ 
and define suitable \Ct{} formulae for them. The exact properties
these formulae are to satisfy, will be specified below.
Next the induction takes care of $s$ (the parent of $s'$)
using the formula already specified for $s'$. Thus we 
obtain a \Ct{} formula for $s$.
Then, using the formula for $s$,
we define a \Ct{} formula for $r$. After this, we
%
%
%
define a \Ct{} formula for the node $w$ using
the formulae defined for $r$ and $x'$; notice that this time (when 
going from $r$ to $w$) the relation $\lessv$ is
scanned in a different direction as in the previous step when going from $s$ to $r$.
Thus we have now defined the formulae
for each of the neighbours of $x$ in the diagram $M_{\psi}$, \ie, 
the formulae for $w$ and $y$. Therefore we can finally define the formula for $x$ (scanning
the relation $\succv$ both ways, upwards from $y$ and downwards from $w$).
Notice that since $\varphi(x)$ is in rooted diagram normal form, $M_{\psi}$ is
connected, which in the this example is due to the node $r$.
Thus the induction does not have to deal with any disconnected components.
If $u$ is a point in the diagram $M_{\psi}$, then we 
let $M_{\psi}(u)$ denote the substructure of $M_{\psi}$ 
induced by the set that contains $u$ and all the nodes 
occurring before $u$ in the induction from the leaves towards $x$.
Formally, the domain of $M_{\psi}(u)$ consists of the points $u'\in M_{\psi}$
such that the undirected path in $M_{\psi}$ from $x$ to $u'$ contains $u$.
We next show how to inductively define, for each node $v\in M_{\psi}$, a
formula $\chi_v(y)$ of \Ct{}
(with the sole free variable $y$)
such that the following condition holds.

\medskip

\noindent
Let $N$ be a $\sigma$-tree and $v'$ a point in $N$.
Then $N\models\chi_v(v')$ iff $N$ contains an induced 
substructure isomorphic to $M_{\psi}(v)$ via an
isomorphism that sends $v'$ to $v$.

\medskip
Let $v\not= x$ be a leaf of $M_{\psi}$. 
Then we let $\chi_v(y)$ be the formula $\delta_v(y)$ (defined above).
Now assume $v$ is not a leaf of $M_{\psi}$.
In order to show how to define $\chi_v(y)$, we first give some auxiliary definitions.
\begin{enumerate}
\item
Let $c_1,...\, ,c_{n_+}$ denote the children of $v$ in $M_{\psi}(v)$, \ie, the
elements $c_i$ in $M_{\psi}(v)$ such that $(v,c_i)\in\, \succv$.
(We note that possibly $n_+ = 0$.)
\item
Let $c'$ denote a possible parent of $v$
in $M_{\psi}(v)$, \ie, a point $c'$ in $M_{\psi}(v)$ such that $(c',v)\in\succv$. (We note that
possibly there exists no such point $c'$.)
\item
Let $d_1,...\, ,d_{m_+}$ denote the descendants $d_i$ of $v$ in $M_{\psi}(v)$ 
such that
\begin{enumerate}
\item
$(v,d_i)\in\ \lessv \setminus\succv$, \ie, $d_i$ is a 
descendant of $v$ but not a child of $v$,
\item
there exists no point $d$ in
the diagram $M_{\psi}$ such that $v\lessv d\lessv d_i$.
%
%
%
%
%
%
\end{enumerate}
We note that possibly $m_+ = 0$.
\item
Let $d'$ denote the (possibly non-existing) ancestor of $v$ in $M_{\psi}(v)$
such that $d'$ is \emph{not} a parent of $v$ and
there does not exist a point $d$ in the diagram $M_{\psi}$ such
that $d'\lessv d\lessv v$.
Note that if we have the 
point $c'$ (specified in bullet 2 above) in the diagram, then there is no $d'$ in the diagram.
Vice versa, if there is a $d'$ in the diagram, there is no $c'$.
Also, it is possible that neither $c'$ nor $d'$ exists in the diagram.
\end{enumerate}

\smallskip

By the induction hypothesis, we
have defined all the required formulae denoted by
$$\chi_{c_1},...,\chi_{c_{n_+}},\chi_{c'},
\chi_{d_1},...,\chi_{d_{m_+}},\chi_{d'}$$
each with the free variable $y$. 
(We note that $\chi_{c'}$ and $\chi_{d'}$
cannot not \emph{both} be in the list of required formulae.)
%
%
%
%
%
Now, we shall next consider \emph{collections} of these formulae defined as follows. A
collection (in the variable $y$) is here defined to be a conjunction that contains, for each of the
formulae $\chi(y)$ defined in 
the previous stage of the induction (\ie, some subset of the above listed formulae $\chi_{c_1},...,\chi_{c_{n_+}},\chi_{c'},
\chi_{d_1},...,\chi_{d_{m_+}},\chi_{d'}$), 
exactly one of the formulae $\chi(y),\neg\chi(y)$ as a conjunct.
(We note that a collection may be unsatisfiable if 
for example $\chi_{c_1}$ and $\chi_{c_2}$ are equivalent and we
include the negation of exactly one of them in the collection. Note also that a node 
satisfying some formula $\chi_{v_1}(y)$ can also satisfy a non-equivalent formula $\chi_{v_2}(y)$ as 
the formulae are only supposed to assert that a certain substructure exists `below' $y$ in the undirected tree
with root $x$.)
Our next step is to show how $\chi_v(y)$ can be defined based on the
corresponding formulae $\chi_{c_1},...,\chi_{c_{n_+}},\chi_{c'},
\chi_{d_1},...,\chi_{d_{m_+}},\chi_{d'}$ for the nodes
that occur immediately before $v$ in the induction.
We illustrate how this is done via an elucidating example.
Consider a situation where $n_+ = 2$, ${m_+} = 1$ and
neither $c'$ nor $d'$ exists.\footnote{For simplicity,  we assume here
and in all subsequent examples that the
set of unary predicates considered is empty.}
The formula $\chi_v(y)$ will be a big disjunction
over all possible suitable scenarios that the nodes in the previous stage 
could realize in some (any) tree with the desired substructure.
%
%
%
%
One possible scenario in some structure is the one
where one child of the current node $v$ satisfies $\chi_{c_1}(y)\wedge\chi_{c_2}(y)$,
another child satisfies $\neg \chi_{c_1}(y)\wedge\chi_{c_2}(y)$, and a
third child satisfies $\chi_{c_1}(y)\wedge\neg\chi_{c_2}(y)\wedge
\exists x(\lessv(y,x) \wedge\chi_{d_1}(x))$.
In this scenario there must exist two distinct children such
that one satisfies $\chi_{c_1}$ and the other one $\chi_{c_2}$,
and furthermore, there must be a third child that connects via $\lessv$ to a
node satisfying $\chi_{d_1}$.
This kind of a condition is
easily expressible in \Ct{}.
Another suitable scenario is that three or more
children all satisfy the formula $$\alpha(y) := \chi_{c_1}(y)\wedge\chi_{c_2}(y)\wedge
\exists x(\lessv(y,x)\wedge\chi_{d_1}(x)),$$
as this ensures that there is a child $u_1$ that satisfies $\chi_{c_1}(y)$ and 
another child $u_2$ that satisfies $\chi_{c_2}(y)$, and furthermore, a
descendant (which is neither a child nor reachable via $u_1$ or $u_2$) that satisfies $\chi_{d_1}$.
Here all the three children 
satisfy the same collection $\chi_{c_1}(y)\wedge\chi_{c_2}(y)$.
Again it is easy to describe this scenario in $\mathrm{C}^2$ by 
stating the existence of at least three children satisfying $\alpha(y)$. It is
not difficult to see how to write the full disjunction $\chi_v(y)$ that 
covers all the possible suitable scenarios and thereby enumerates the
ways to connect $v$ to the 
nodes in the previous stage.
We consider one more example. This time we assume
that $m_+ = n_+ = 0$ and that $c'$ exists.
Furthermore, we assume that $c'$ has two children in 
the diagram, $v$ and $u$, and the induction proceeds from $u$ via $c'$ to $v$.
Now note that we cannot simply define $\chi_v(y)$ to
be the formula $\exists x(x\succv y\wedge\chi_{c'}(x))$,
but instead, we must define $\chi_v(y)$ to be the formla 
$$(\neg\chi_u(y)\wedge\exists x(x\succv y\wedge\chi_{c'}(x)))\ \vee
(\chi_u(y)\wedge\exists x(x\succv y\wedge\chi_{c'}(x)
\wedge\exists^{\geq 2}y(x\succv y \wedge \chi_u(y)))).$$

We omit further details since it is now easy to see in general how to write the
formulae $\chi_v(y)$ in \Ct{}
using collections and \emph{counting}; one simply enumerates all possible situations 
with sufficient numbers of correctly oriented neigbours satisfying the
formulae $\chi_{c_1},...,\chi_{c_{m_+}},
\chi_{c'},\chi_{d_1},...,\chi_{d_{n_+}},\chi_{d'}$.
Then the big disjunction over all the resulting possible ways to connect $v$ to the
nodes in the previous stage is the desired formula.
\end{proof}
\begin{cor}\label{diagramsfi}
Let $\varphi = \exists\overline{y}\psi$ be a
formula of \ODF$[\succv, \lessv]$. The formula $\varphi$ may be a 
sentence or contain a single free variable. Then there
exists a formula $\varphi^*$ of \Ct$[\succv, \lessv]$ that is equivalent to $\varphi$.
%
%
%
%
%
\end{cor}
\begin{proof}
The case where $\varphi$ has a single free variable follows from
the previous lemma. The case where $\varphi$ is a
sentence is covered as follows.
Assume $\varphi = \exists\overline{x}\psi$. Remove a
single variable from $\overline{x}$ and apply the argument
for the case with a free variable. Then reintroduce a quantifier that
quantifies the remaining free variable away.
\end{proof}

\section{Satisfiability of one-dimensional fragment over trees} \label{s:sattrees}

The aim of this section is to establish the complexity of  satisfiability of \ODF{} over trees for all navigational signatures contained in
$\{ \succv, \lessv, \succh, \lessh \}$.
En route we  show some small model properties, allowing us, when designing algorithms deciding  satisfiability, to restrict attention to models with
appropriately bounded vertical and horizontal paths. 

Actually, we can show that when the child relation is present in the signature then the satisfiability
problem is \TwoExpTime-complete using some known results on \UNFO.

\begin{thm} \label{t:treefull}
Let $\{ \succv \} \subseteq \sigma_{nav} \subseteq \{ \succv, \lessv, \succh, \lessh \}$. Then the satisfiability problem
for \ODF$[\sigma_{nav}]$ is \TwoExpTime-complete.
\end{thm}
\begin{proof}
In \cite{SC13} a \TwoExpTime-upper bound is given for \UNFO$[\succv, \lessv, \succh, \lessh]$, and the corresponding lower bound---for
\UNFO$[\succv]$. We transfer these bound to \ODF{} using  the polynomial translation from \ODF$[\succv, \lessv, \succh, \lessh]$ to 
\UNFO$[\succv, \lessv, \succh, \lessh]$ in the proof of Thm.~\ref{t:expr1} and, respectively, Lemma \ref{l:unfotoodf}.
\end{proof}

We will soon see that for the  navigational signatures not containing $\succv$ but containing $\lessv$ the complexity drops down to \ExpSpace.
 The machinery we will develop will also allow us to give an alternative,
direct proof of the upper bound in  Thm.~\ref{t:treefull}

\subsection{Profiles for trees}
We begin with an adaptation of the notion of profiles for the case of trees.
As in the case of words, the \emph{profile} of a node says what the types
of all tuples (of some bounded size) containing this node are.

\begin{figure}
\begin{tikzpicture}[scale=0.8]
	\foreach \x/\y in {0/0, 1/0, 2/0, 3/0, 4/0, 5/0,
	                   1/1, 2/1, 3/1, 4/1,
										 1/2, 2/2, 3/2, 4/2, 5/2,
										 2/3, 3/3, 4/3, 5/3,
										4/4,
										1/-1}{
  	  \filldraw[fill=white] (\x, \y) circle (0.08);
			}
			
	\draw[->] (1-0.1,1-0.1) -- (0+0.1,0+0.1);
	\draw[->] (2-0.1,1-0.1) -- (1+0.1,0+0.1);
	\draw[->] (4-0.1,1-0.1) -- (3+0.1,0+0.1);
	\draw[->] (3-0.2,2-0.1) -- (1+0.1,1+0.1);
	\draw[->] (3-0.1,2-0.1) -- (2+0.1,1+0.1);
	\draw[->] (2-0.1,3-0.1) -- (1+0.1,2+0.1);
	\draw[->] (4-0.1,3-0.1) -- (3+0.1,2+0.1);
	\draw[->] (4-0.2,4-0.1) -- (2+0.1,3+0.1);
	\draw[->] (4-0.1,4-0.1) -- (3+0.1,3+0.1);
	\draw[->] (2-0.1,0-0.1) -- (1+0.1,-1+0.1);
	
	\draw[->] (2,1-0.1) -- (2,0+0.1);
	\draw[->] (3,2-0.1) -- (3,1+0.1);
	\draw[->] (4,1-0.1) -- (4,0+0.1);
	\draw[->] (2,3-0.1) -- (2,2+0.1);
	\draw[->] (4,3-0.1) -- (4,2+0.1);
	\draw[->] (4,4-0.1) -- (4,3+0.1);
	
	\draw[->] (4+0.1,1-0.1) -- (5-0.1,0+0.1);
	\draw[->] (3+0.1,2-0.1) -- (4-0.1,1+0.1);
	\draw[->] (4+0.1,3-0.1) -- (5-0.1,2+0.1);
	\draw[->] (4+0.1,4-0.1) -- (5-0.1,3+0.1);
	
	
	\draw[->] (1+0.1,0) -- (2-0.1,0);
	\draw[->] (3+0.1,0) -- (4-0.1,0);
	\draw[->] (4+0.1,0) -- (5-0.1,0);
	\draw[->] (1+0.1,1) -- (2-0.1,1);
	\draw[->] (2+0.1,1) -- (3-0.1,1);
	\draw[->] (3+0.1,1) -- (4-0.1,1);
	\draw[->] (1+0.1,2) -- (2-0.1,2);
	\draw[->] (3+0.1,2) -- (4-0.1,2);
	\draw[->] (4+0.1,2) -- (5-0.1,2);
	\draw[->] (2+0.1,3) -- (3-0.1,3);
	\draw[->] (3+0.1,3) -- (4-0.1,3);
	\draw[->] (4+0.1,3) -- (5-0.1,3);

	\node at (2.25,0.75) {$a$};
	\filldraw[fill=black] (2, 1) circle (0.08);
	
	
	\fill [gray,rounded corners=3, fill opacity=0.2, draw]
  (0,-0.25) --
  (-0.25,0) --
  (1,1.25) --
  (1.25,1) --
  cycle
  {};
	
	\fill [gray,rounded corners=3, fill opacity=0.2, draw]
  (1,-1.25) --
  (0.75,-1) --
	(0.75, 0.25) --
	(2, 0.25) -- 
	(2.25, 0) --
  cycle
  {};
	
	\fill [gray,rounded corners=3, fill opacity=0.2, draw]
	(3,-0.25) --
	(2.75,0) --
  (2.75,1) --
	(3,1.25) --
	(4,1.25) --
  (5.25, 0) --
	(5,-0.25) --
  cycle
  {};
	
	\fill [gray,rounded corners=3, fill opacity=0.2, draw]
	(1,1.75) --
	(0.75,2) --
	(2,3.25) --
	(4,4.25) --
	(5.25,3) --
	(5.25,1.75) --
	cycle
	{};
	
	\node at (0.2,0.8) {$L$};
	\node at (4.8,0.8) {$R$};
	\node at (1.8,-0.8) {$B$};
	\node at (2.7,3.9) {$A$};
	
\end{tikzpicture}
\caption{Positions in a tree with respect to a node $a$.}%
\label{f:positions}%
\end{figure}
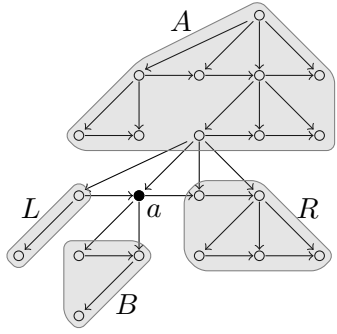

Given a tree $\str{T}$ and its nodes $a, b$, $b \not=a$, we say that  $b$ is \emph{in position} $B$ to $a$ (or \emph{below} $a$) if 
it belongs to the subtree of $a$, \emph{in position} $L$ to $a$ (or \emph{left} to $a$) if it belongs to the subtree of some left sibling of $a$, 
\emph{in position} $R$ to $a$ (or \emph{right} to $a$) if it belongs to the subtree of some right sibling of $a$, and \emph{in position} $A$ to $a$
(or \emph{above} $a$) if it is not in any of the previous positions to $a$.
See Fig.~\ref{f:positions}.

Let $\sigma=\sigma_0 \cup \sigma_{nav}$ be a signature such that $\lessh \in \sigma_{nav}$\footnote{This assumption is of technical character, and our approach could
be also developed for signatures not containing $\lessh$.} and let  $\str{T}$ be a tree over this signature. We say that an element $a \in T$  \emph{realizes} (or \emph{has}) a $k$-$\sigma$-\emph{profile} $\sigma$-$\prof{{T}}{k}{a}=(\cF, \cA, \cB, \cL, \cR)$ if $\cF$ is the set of all $s$-types, $1 \le s \le k$, realized by tuples $a_1, \ldots, a_s$ such
that $a_1=a$, and for any  position $P \in \{A, B, L, R \}$, the component $\cP$
is the subset of $\cF$ consisting of the types realized by those tuples for which  for all $2 \le i \le k$ the element $a_i$ is in position $P$ to $a$. When $\sigma$ is clear from context we will just speak 
about $k$-profiles and write 
$\prof{{T}}{k}{a}$ instead of $k$-$\sigma$-profiles and $\sigma$-$\prof{{T}}{k}{a}$.
Given a $\sigma$-$k$-profile $\theta$ we will sometimes refer to its components with $\theta.\cF$, $\theta.\cA$, $\theta.\cB$, $\theta.\cL$ and $\theta.\cR$. 

\begin{lem} \label{l:computef}
Let $\theta$ be a profile of a node in a tree over a signature containing $\lessh$. Then
$\theta.\cF$ is unequivocally determined by $\theta.\cA$, $\theta.\cB$, $\theta.\cL$ and $\theta.\cR$
Moreover, there is a procedure ${\tt fulltype}(\cA, \cB, \cL, \cR)$ which given $\theta.\cA$, $\theta.\cL$, $\theta.\cB$ and $\theta.\cR$  computes $\theta.\cF$ in time polynomial in $|\theta|$
and exponential in $k$.
\end{lem} 
\begin{proof}
The unique determination of $\cF$ for an element $a$ follows from the fact that for any elements $b_1, b_2$, if we know their positions
to $a$ and the truth values of the atoms $a \rightleftharpoons b_1$, $b_1 \rightleftharpoons a$, $a \rightleftharpoons b_2$, $b_2 \rightleftharpoons a$
for all $\rightleftharpoons \in \sigma_{nav}$ then the truth values of $b_1 \rightleftharpoons b_2$ and $b_2 \rightleftharpoons b_1$ are determined
for all $\rightleftharpoons$.\footnote{This could be not true if  $\lessh \not\in \sigma_{nav}$: assuming that $b_1$ is the parent of $a$
and $b_2$ is in position $R$ to $a$ but not joined to it by any relation then we do not know if $b_1$ is the parent of $b_2$ or just an ancestor.}

More specifically, to construct all types in $\cF$ we proceed as follows. Construct all possible tuples consisting of 
at most one type from each of the components $\theta.\cA$, $\theta.\cB$, $\theta.\cL$ and $\theta.\cR$
and for each such tuple combine its types together into a single type in a natural way, that is identify their  $x_1$ variables,
appropriately renumber the other variables, and appropriately set the navigational  relations. 
The latter is done using the observation that the only navigational connections between elements in different positions to an 
element $a$ are as follows:
\begin{itemize}
\item if $b$ is in position $B$ (to $a$) then $c \lessv b$ holds for those $c$ in position $A$
for which $c \lessv a$ holds.
\item  if $b$ is in position $L$ then $c \lessv b$ holds for those $c$ in position $A$ for which $c \lessv a$ holds,
$c \succv b$ holds if $c \succv a$ holds, 
and $b \lessh d$ holds for those $d$ in position $R$ for which $a \lessh d$ holds,
\item symmetrically for $b$ in position $R$.
\end{itemize}
 
If the number $l$ of variables in the so obtained type $\pi$ is not greater than $k$ then the procedure adds to $\cF$ the type $\pi$, together
with all the types obtained from $\pi$ by permuting its  variables $x_2, \ldots, x_l$. 
\end{proof}

For example let us consider the signature $\sigma$ with $\sigma_0=\{P \}$ and $\sigma_{nav}= \{\lessv, \lessh \}$
and assume that the $2$-type $\pi_1=\{Px_1, Px_2, x_2 \lessv x_1 \}  \in \cA$,  and the $3$-type $\pi_2=\{Px_1, Px_3, x_3 \lessv x_2 \}$ $\in \cR$
(we list only non-negated literals in the types). The combination of $\pi_1$ and $\pi_2$ is the $4$-type 
$\pi_3=\{P x_1, P x_2, P x_4, x_2 \lessv x_1, x_2 \lessv x_3, x_2 \lessv x_4, x_4 \lessv x_3 \}$. 
See Fig.~\ref{f:joiningtypes}. The type $\pi_3$ together with the types
obtained from $\pi_3$ by permuting the variables $x_2, x_3, x_4$ in all possible ways are added to $\cF$ if $k \ge 4$.

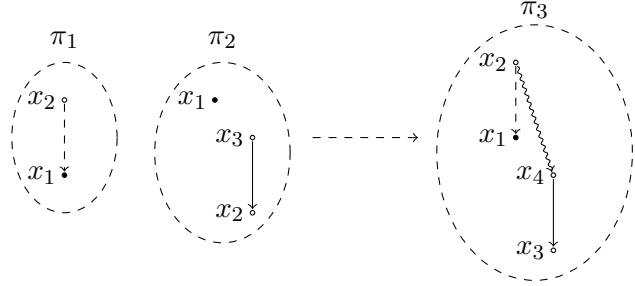
\begin{figure}
\begin{tikzpicture}

\draw[->, dashed]  (1,2-0.05) -- (1,1+0.05);
\draw[fill=black] (1,1) circle (0.03);
\draw (1,2) circle (0.03);

\draw[->] (3.5,1.5-0.05) -- (3.5,0.5+0.05);
\draw[fill=black] (3,2) circle (0.03);
\draw (3.5,1.5) circle (0.03);
\draw (3.5,0.5) circle (0.03);

\draw (7,2.5) circle (0.03);
\draw[fill=black] (7,1.5) circle (0.03);
\draw (7.5,1) circle (0.03);
\draw (7.5,0) circle (0.03);

\draw[->, dashed] (7,2.5-0.05) -- (7,1.5+0.05);
\draw[->,decorate,decoration={snake,amplitude=.2mm,segment length=1mm}] (7+0.02,2.5-0.05) -- (7.5-0.02,1+0.05);
\draw[->] (7.5,1-0.05) -- (7.5,0+0.05); 

\node at (0.7,1) {$x_1$};
\node at (0.7,2) {$x_2$};

\node at (2.7,2) {$x_1$};
\node at (3.2,1.5) {$x_3$};
\node at (3.2,0.5) {$x_2$};

\node at (6.7, 2.5) {$x_2$};
\node at (7.2, 1) {$x_4$};
\node at (7.2, 0) {$x_3$};
\node at (6.7,1.5) {$x_1$};

\node at (1,2.8) {$\pi_1$};
\node at (3.1,2.8) {$\pi_2$};
\node at (7.25,3.2) {$\pi_3$};

\draw [->, dashed] (4.3,1.5) -- (5.7,1.5);

\draw[dashed] (1,1.5) ellipse (0.7 and 1);
\draw[dashed] (3.1,1.3) ellipse (0.9 and 1.2);
\draw[dashed] (7.25,1.3) ellipse (1.3 and 1.7);

\end{tikzpicture}
\caption{Combining type $\pi_1$, $\pi_2$ into a single type $\pi_3$. The dashed and solid connection in $\pi_3$ are present in $\pi_1$ and, resp., $\pi_2$,
the wavy connection follows from the definition of profiles.} \label{f:joiningtypes}
\end{figure}

\begin{lem} \label{l:profilestrees}
The number of $k$-$\sigma$-profiles 
is bounded from above 
by $\fg^\star(|\sigma_0|, k)$ where $\fg^\star: \N \times \N \rightarrow \N$ is a fixed function, doubly exponential  in its both arguments. 
\end{lem}
\begin{proof}
Each component of a profile is determined by the set of the $k$-types it contains.
The number of $k$-types in a component can be roughly estimated by $(2^{|\sigma_0|})^k \cdot 9  k (k-1)$
(the number of possible assignments of $1$-types to the elements of a tuple of $k$ elements, times the number
of possible connections by relations from $\sigma_{nav}$ for a pair of elements $a, b$: $a$ is equal to $b$, $b$ is the next sibling of $a$, $b$ is a following sibling of $a$ but not
the next one, $b$ is a child of $a$, $b$ is a descendant of $a$ but not its child; and vice versa).
So, each component has at most $2^{(2^{|\sigma_0|})^k \cdot 9  k (k-1)}$ possible values.
From this the existence of the desired $\fg^\star$ follows. 
\end{proof}

In contrast to the case of words, where in a single model at most exponentially many profiles are realized
(Lemma \ref{l:profiles}), in the case of trees
there are \ODF{} formulas whose models must realize doubly exponentially
many profiles. In the next subsection we will however see, that for some signatures we can at least
bound exponentially the number of profiles on all the vertical or horizontal paths in "minimal" models of any formula,
which will allow us to prove our small model properties.

Again, we connect the notion of profiles with satisfaction of normal form formulas.
Given a normal form formula $\phi$ of width $k$ over a signature $\sigma$ we say that a $k$-$\sigma$-profile $\theta$
is \emph{compatible} with $\phi$ if 
\begin{itemize}
\item for every conjunct
$\forall x_1 \ldots x_{l_i} \phiu_i(x_1 \ldots x_{l_i})$ of $\phi$ and every $l_i$-type $\pi \in \theta.\cF$, we
have $\pi \models \phiu_i$.
\item for every conjunct $\forall y_0 \exists y_1 \ldots y_{k_i} \phie_i(y_0, y_1 \ldots y_{k_i})$ of $\phi$
there is a $(k_i+1)$-type  $\pi \in \theta.\cF$ such that
$\pi \models \phie_i(x_1, \ldots, x_{k_i+1})$.
\end{itemize}

It is straightforward to see:
\begin{lem} \label{l:tcompatible}
A normal form formula $\phi$ of width $k$ over a signature $\sigma$  is satisfied in a tree 
$\str{T}$ iff every $k$-$\sigma$-profile realized in $\str{T}$ is compatible with $\phi$.
\end{lem}

For a further use we make the following observation.

\begin{lem} \label{l:uniqdeter}
Let $\str{T}$ be a $\sigma$-tree, $a,a'$ its nodes and $\theta$, $\theta'$ their respective $k$-$\sigma$-profiles. 
\begin{romanenumerate}
\item If $a'$ is a child of $a$ then 
 \begin{alphaenumerate}
   \item $\theta'.\cA$ is uniquely determined by  the $1$-type of $a$ and $\theta.\cL$, $\theta.\cR$ and $\theta.\cA$
	\item $\theta.\cB$ is uniquely determined by the $1$-type of $a'$ and $\theta'.\cL$, $\theta'.\cR$ and $\theta.\cB$

	\end{alphaenumerate}
\item If $a'$ is the next sibling of $a$ then
  \begin{alphaenumerate}
 \item $\theta'.\cL$ is uniquely determined by the $1$-type of $a$ and $\theta.\cL$ and $\theta.\cB$.
	\item $\theta.\cR$ is uniquely determined by the $1$-type of $a'$ and $\theta'.\cR$, $\theta'.\cB$.
	\end{alphaenumerate}

\end{romanenumerate}

Moreover there is a procedure $\tt{computeA}(\mu, \cL, \cA, \cR)$ which computes the component  $\theta'.\cA$ when  given  the $1$-type of $a$ and $\theta.\cL$, $\theta.\cR$ and $\theta.\cA$, in time polynomial in $|\theta|$ and exponential in $k$. Analogously there are procedures
   $\tt{computeB}(\mu, \cL, \cR, \cB)$, $\tt{computeL}(\mu, \cL, \cB)$ and $\tt{computeR}(\mu, \cR, \cB)$ computing  $\theta.\cB$, $\theta'.\cL$, $\theta.\cR$,
	respectively, when fed with the appropriate parameters (as in points (i)(b), (ii)(a), (ii)(b)). \end{lem}
\begin{proof} 

Consider the statement (i)(a). 
Note that the set of nodes in position $A$ to $a'$ consists precisely 
of $a$ and the elements in positions $L$, $R$ and $A$ to $a$. 
Thus the procedure $\tt{computeA}$ can work as follows. 

 Construct all possible tuples consisting of at most one type from each of the components $\theta.\cL$, $\theta.\cR$  and $\theta.\cA$
and for each such tuple combine its types together into a single type $\pi$ in a natural way, that is identify their  $x_1$ variables,
appropriately renumber the other variables, and appropriately set the navigational  relations (similarly as it was done 
in the proof of Lemma \ref{l:computef}). Construct $\pi'$ by increasing the number of every variable in $\pi$ by $1$ and adding $x_1$ 
as
''a child of $1$-type $\mu$ of (the current) $x_2$'' (that is by setting the truth of $\sigma_0$-atoms containing $x_1$ in accordance with $\mu$ and appropriately setting the truth of $\sigma_{nav}$-atoms containing $x_1$ and the other $x_i$). 
Then construct $\pi''$ from $\pi'$ by removing all the literals that contain $x_2$ and then decreasing the number
of each variable $x_i$, $i \ge 2$ by $1$. 
If the number $l$ of variables in the so obtained type $\pi'$ ($\pi''$) is not greater than $k$ then add to $\cF$ the type $\pi'$ ($\pi''$), together
with all the types obtained from $\pi'$ ($\pi''$) by permuting its  variables $x_2, \ldots, x_l$.

The other statements can be justified analogously.
\end{proof}

We say that an $s$-type if \emph{trivial}  if for every $1 \le i, j \le s$ it contains $x_i=x_j$, that is it is realized only by singletons.
We say that a trivial $s$-type is \emph{based on  $1$-type} $\mu$ if its restriction to $x_1$ is equal to $\mu$.
We now
define the following notion of
local consistency.

\begin{defi} \label{d:localcons}
Let $\str{T}$ be a tree, $k$ a natural number, $\Omega$ a function assigning to each $a\in T$ a $1$-type and $\Xi$ a function assigning to each $a \in T$ a tuple $(\cF, \cA, \cB, \cL, \cR)$ of collections of $s$-types such that $\cF=\tt{fulltype}(\cA, \cB, \cL, \cR)$, $s \le k$.
We say that the pair $(\Omega, \Xi)$ is \emph{locally consistent} on $\str{T}$ if the following conditions hold.
\begin{enumerate}[(a)]
\item if $a \in T$ is the root then $\Xi(a).\cA$ is trivial and based on $\Omega(a)$. 
\item if $a \in T$ is a leaf then $\Xi(a).\cB$ is trivial and based on $\Omega(a)$.
\item if $a \in T$ has no preceding sibling then $\Xi(a).\cL$ is trivial and based on $\Omega(a)$.
\item if $a \in T$ has no following sibling then $\Xi(a).\cR$ is trivial and based on $\Omega(a)$.
\item for any $a,a' \in T$ such that $a'$ is a child of $a$ we have
\begin{enumerate}[(i)]
\item $\Xi(a').\cA = \tt{computeA}(\Omega(a'), \Xi(a).\cL, \Xi(a).\cR, \Xi(a).\cA)$  
\item $\Xi(a).\cB = \tt{computeB}(\Omega(a), \Xi(a').\cL, \Xi(a').\cR, \Xi(a').\cB)$
\end{enumerate}
\item for any $a,a' \in T$ such that $a'$ is the next sibling of $a$ we have
\begin{enumerate}[(i)]
\item $\Xi(a').\cL=\tt{computeL}(\Omega(a'), \Xi(a).\cL, \Xi(a).\cB)$  
\item $\Xi(a).\cR=\tt{computeR}(\Omega(a), \Xi(a).\cR, \Xi(a).\cB)$
\end{enumerate}

\end{enumerate} 
\end{defi}

Obviously, if $\Omega$ returns the $1$-types of elements of $\str{T}$ and $\Xi$ returns their $k$-profiles then the pair $(\Omega, \Xi)$ is locally consistent. In the following lemma we show that
the opposite is also true. 
\begin{lem} \label{l:localcons}
Let $\str{T}$ be a tree, $k$ a natural number, $\Omega$ the function assigning to each $a \in T$ the $1$-type of $a$ in $\str{T}$, and $\Xi$ a function assigning to each $a \in T$ a tuple $(\cF, \cA, \cB, \cL, \cR)$ of collections of $s$-types, $1 \le s \le k$.
If the pair $(\Omega, \Xi)$ is locally consistent on $\str{T}$ then for every $a \in T$ we have that  $\Xi(a)=\prof{T}{k}{a}$.
\end{lem}

\begin{proof}
Let us first see that  for every $a \in T$ the equality holds for the $\cL$-, $\cB$- and $\cR$-components
of $\Xi(a)$ and $\prof{T}{k}{a}$. Let $d$ be the maximal number
of edges on a vertical path in $\str{T}$.
Define $level(a)$ to be $d$ if $a$ is the root and $level(b)-1$, where $b$ is the parent of  $a$, otherwise.
This way $0 \le level(a) \le d$ for any $a \in T$.
Define $posl(a)$ to be $0$ is $a$ the leftmost child of some node and $posl(b)+1$, where $b$ is the previous sibling of $a$,
otherwise.
Similarly, define $posr(a)$ to be $0$ is $a$ the rightmost child of some node and $posr(b)+1$, where $b$ is the next sibling of $a$,
otherwise. 

 We proceed by induction on the level of a node. For the base of induction assume $level(a)=0$. In this case $a$ is a leaf.  
Then the equality  $\Xi(a).\cB=\prof{\str{T}}{k}{a}.\cB$ follows from Condition (b) of Def.~\ref{d:localcons}.
Consider now the $\cL$-components. We proceed by subinduction on $posl(a)$. If $pos(a)=0$ 
 then the equality $\Xi(a).\cL=\prof{\str{T}}{k}{a}.\cL$ follows from Condition (c). Otherwise, assume that 
for the previous sibling $a'$ of $a$ we have $\Xi(b).\cL=\prof{T}{k}{a'}.\cL$. As it must be that $level(a')=0$
it again follows from Condition (b) that $\Xi(a').\cB=\prof{T}{k}{a'}.\cB$.  
The equality $\Xi(a).\cL=\prof{\str{T}}{k}{a}.\cL$ follows now from Condition (f)(i).
The argument for the $\cR$-components is strictly symmetric, by subinduction on $posr(a)$ and involves Condition (f)(ii).

Assume now that $level(a)=s$, for some $s>0$  and that for any node $b$ with $level(b)=s-1$ we have 
$\Xi(b).\cB=\prof{T}{k}{b}.\cB$, $\Xi(b).\cL=\prof{T}{k}{b}.\cL$, $\Xi(b).\cR=\prof{T}{k}{b}.\cR$ (the main inductive assumption).
This inductive assumption in particular holds for any child of $a$. Again, we first consider the $\cB$-components. If
$a$ is a leaf then the equality  $\Xi(a).\cB=\prof{\str{T}}{k}{a}.\cB$ follows from Condition (b). Otherwise
let $a'$ be a child of $a$. The equality for the $\cB$-components for $a$ follows in this case from the inductive assumption for $a'$ 
and Condition (e)(ii). So, all the nodes on level $s$ have proper $\cB$-components. 
For the $\cL$- and $\cR$-components we can now proceed as in the base of induction.

This finishes the part of the proof concerning the $\cB$-, $\cL$- and $\cR$-components.
It remains to show the equality for the $\cA$-components. This is done by induction on $depth(a)$.
If $depth(a)=0$ ($a$ is the root) then the equality for the $\cA$-components follows from (a). Otherwise let $a'$ be the parent of $a$
and assume that the equality  $\Xi(a').\cA=\prof{\str{T}}{k}{a'}.\cA$ holds. As we have already proved, for all nodes of $\str{T}$
 this equality holds for the $\cL$- and $\cR$-components, so we can use Condition (e)(i) to get that $\Xi(a).\cA=\prof{\str{T}}{k}{a}.\cA$.

Now the equality of the $\cF$-components follows from the fact that they are computed by the procedure $\tt{ComputeF}$ from Lemma \ref{l:computef}.
This finishes the proof.
\end{proof}

\subsection{Size of models} \label{s:sizeofmodels}

In this subsection we show essentially optimal bounds for lengths of vertical and horizontal paths in "minimal" models of normal form formulas, for 
all relevant navigational signatures.

What is crucial for the lower complexity bound in Thm.~\ref{t:treefull} is the ability to enforce  doubly-exponentially long vertical  paths. 
Let us see how to
do it directly in \ODF$[\succv]$. We use unary
predicates $N, P, P_0, \ldots, P_{n-1}, Q$. See the left part of Fig.~\ref{f:longpath}. The intended long path is the path of elements in $N$. Every element in $N$ is going to have
$2^n$ children marked by $P$, each of which has a \emph{local position} in the range $[0, 2^n-1]$ encoded by means of $P_0, \ldots, P_{n-1}$. Reading the truth-values
of $Q$ as binary digits we can assume that the collection of the $P$-children of a node in $N$ encodes its \emph{global position} in the tree in the range $[0,2^{2^n}-1]$ (the $i$-th
bit of this global position is $1$ iff at the element at local position $i$ the value of $Q$ is true). It is then possible to say that each node in $n$ whose global position is smaller than $2^{2^n}-1$ has
a child in $N$ with the global position greater by $1$.

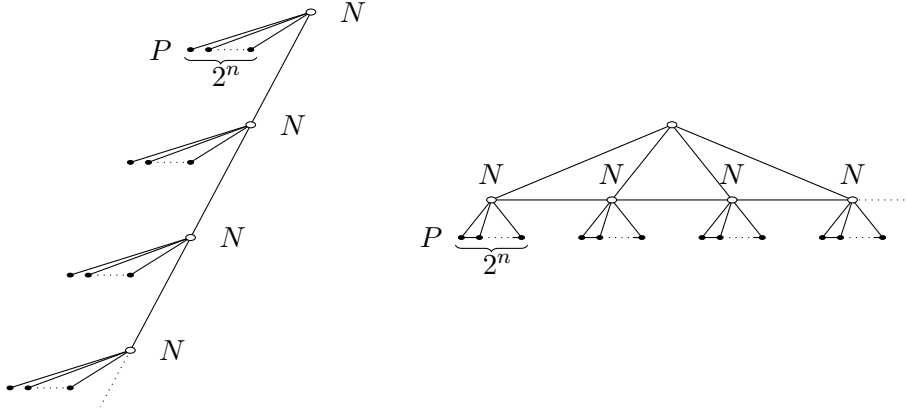
\begin{figure}[h]
\begin{center}
\begin{tikzpicture}[xscale=0.8, yscale=0.5]

\begin{scope}
\foreach \x/\y in  {3/1,4/4,5/7} {
  \draw (\x,\y) -- (\x+1, \y+3);
	}

\draw[dotted] (3,1) -- (2.5,-0.5);

\foreach \x/\y in  {3/1,4/4,5/7,6/10} {
 
	\filldraw[fill=black](\x-1,\y-1)  circle (0.05);
		\filldraw[fill=black](\x-2,\y-1)  circle (0.05);
		\filldraw[fill=black](\x-2+0.3,\y-1)  circle (0.05);
		\draw (\x-2,\y-1) -- (\x, \y);
		\draw (\x-2+0.3,\y-1) -- (\x, \y);
		\draw (\x-1,\y-1) -- (\x, \y);
		\draw[dotted] (\x-2+0.45, \y-1) -- (\x - 1- 0.1, \y-1);
		 \filldraw[fill=white] (\x, \y) circle (0.08);
		}

   \coordinate [label=center:${N}$] (A) at (7-0.3,10);
   \coordinate [label=center:${N}$] (A) at (6-0.3,7);
   \coordinate [label=center:${N}$] (A) at (5-0.3,4);
  \coordinate [label=center:${N}$] (A) at (4-0.3,1);

\coordinate [label=center:${P}$] (A) at (3.5,9);

\draw [
    decoration={
        brace,
        mirror,
        raise=0.35cm
    },
    decorate
] (3.9,9.5) -- (5.1,9.5);

\coordinate [label=center:${2^n}$] (A) at (4.6,8.3);
\end{scope}

\begin{scope}[shift={(8,0)}]
\foreach \x/\y in  {1/5,3/5,5/5} {
  \draw (\x,\y) -- (\x+2, \y);
	}

\draw[dotted] (7,5) -- (8,5);

\foreach \x/\y in  {1/5,3/5,5/5,7/5} {
 
   \draw (\x,\y) -- (4,7);
	\filldraw[fill=black](\x+0.5,\y-1)  circle (0.05);
		\filldraw[fill=black](\x-0.5,\y-1)  circle (0.05);
		\filldraw[fill=black](\x-0.5+0.3,\y-1)  circle (0.05);
		\draw (\x-0.5, \y-1) -- (\x-0.5+0.3, \y -1);
		\draw (\x+0.5-0.15, \y-1) -- (\x+0.5, \y -1);
		
		\draw (\x-0.5,\y-1) -- (\x, \y);
		\draw (\x-0.5+0.3,\y-1) -- (\x, \y);
		\draw (\x+0.5,\y-1) -- (\x, \y);
		\draw[dotted] (\x-0.5+0.45, \y-1) -- (\x +0.5- 0.1, \y-1);
		 \filldraw[fill=white] (\x, \y) circle (0.08);
		}

   \coordinate [label=center:${N}$] (A) at (1,5.7);
   \coordinate [label=center:${N}$] (A) at (3,5.7);
   \coordinate [label=center:${N}$] (A) at (5,5.7);
  \coordinate [label=center:${N}$] (A) at (7,5.7);

\coordinate [label=center:${P}$] (A) at (0,4);

\draw [
    decoration={
        brace,
        mirror,
        raise=0.35cm
    },
    decorate
] (0.4,4.5) -- (1.6,4.5);

\filldraw[fill=white] (4, 7) circle (0.08);
\coordinate [label=center:${2^n}$] (A) at (1.1,3.3);
\end{scope}
	
\end{tikzpicture}
\end{center}
\caption{Enforcing doubly exponential vertical path in \ODF$[\succv]$ (left) and horizontal path in \ODF$[\lessv, \succh]$ or \ODF$[\succv, \succh]$ (right).} \label{f:longpath}
\end{figure}

We employ the abbreviation $\lceq(x,y)$ in order to state that $x$ and $y$  have the same local position; 
$\lcless(x,y)$ to state that the local position of $y$ is greater than the local position of $x$; and
$\lcsucc(x,y)$ to state that the local position of $y$ is one greater than the local position of $x$
(addition modulo $2^n$). All these abbreviations can be defined in the standard way using quantifier-free formulas of length polynomial in $n$.
Formulas (\ref{f:pathone})-(\ref{f:pathlast}) take care of the basic shape of models (existence of $N$-successors, $P$-successors, and exponentially
many $P$-siblings with appropriate local positions):
\begin{align}
&\exists x (Nx  \wedge \forall y (x \succv y \wedge Py \rightarrow \neg Qy)) \label{f:pathone}\\
&\forall x (Nx \dot{\vee} Px)  \\
& \forall x (Nx \rightarrow \exists y (x \succv y \wedge Py))\\
&\forall x (Px \rightarrow \exists yz (z \succv x \wedge z \succv y \wedge Py \wedge \lcsucc(x,y))) \\
& \forall xyz (z \succv x \wedge z \succv y \wedge \lceq(x,y) \rightarrow (Qx \leftrightarrow Qy)) \\
&\forall x (Nx \wedge \exists y (x \succv y \wedge Py \wedge \neg Qy) \rightarrow \exists y (x \succv y \wedge Ny)) \label{f:pathlast}
\end{align}

Let $\mu(x)$ abbreviate a formula stating that $x$ is an element for which $Q$ is false, and all its siblings with smaller local position 
have $Q$ true. Now we can naturally encode $+1$ addition in our $2^{2^n}$-global-position-counter:
\begin{align}
& \forall xyx'y'zt \big((z \succv x \wedge z \succv x' \wedge z \succv t \wedge t \succv y \wedge t \succ y' \wedge \nonumber \\ 
&\hspace*{50pt} Nx \wedge Nt \wedge Px \wedge Px' \wedge Py \wedge Py' \wedge \mu(x)
   \lceq(x,y) \wedge \lceq(x', y')) \rightarrow \nonumber \\
& \hspace*{90pt} Q(y) \wedge \lcless(x',x) \rightarrow \neg Qy' \wedge \lcless(x,x') \rightarrow (Qy' \leftrightarrow Qx')\big)
\end{align}

In an analogous way, assuming that $\succh$ is available in the signature, that is in \ODF$[\succv, \succh]$ or 
\ODF$[\lessv, \succh]$, we can enforce a doubly exponentially long horizontal  path, like the path of the $N$ in the right part of Fig.~\ref{f:longpath}. 

\medskip
It turns out that the presence of the successor relation(s) is crucial for enforcing doubly exponentially long vertical or horizontal paths. 
To show this let us prove two contraction lemmas.

\begin{lem}\label{l:tcont1}
Let $\sigma$ be a signature with the navigational part containing $\lessh$. Let $\str{T}$ be a $\sigma$-tree and $a,b,a',b' \in T$ 
be such that $b$ is a child of $a$, $b'$ is a child of $a'$, $a'$ is a descendant of $b$, 
$\sigma$-$\prof{\str{T}}{k}{a}.\cB=\sigma$-$\prof{\str{T}}{k}{a'}.\cB$,
and $\sigma$-$\prof{\str{T}}{k}{b}.\cA=\sigma$-$\prof{\str{T}}{k}{b'}.\cA$.
Let $\str{T}'$ be the tree obtained from $\str{T}$ by replacing the subtree of $a$ by the subtree of $a'$, with the exception of the root of 
this subtree which remains $a$. 
Then, for any node $c \in T'$ we have $\sigma$-$\prof{\str{T}}{k}{c}$=$\sigma$-$\prof{\str{T}'}{k}{c}$.
\end{lem}
\begin{proof}
We consider three cases. In the first two of them we analyse the profiles of the elements lying next to
the cut made in our surgery, in the third one we systematically analyse the profiles of the remaining elements of $\str{T}'$.

\medskip \noindent
(i) Assume first that $c$ is a child of $a$ in $\str{T}'$ (that is, it is a child of $a'$ in $\str{T}$).
Clearly the $\cL$-, $\cB$- and $\cR$-component of $\prof{\str{T}}{k}{c}$ are retained in $\str{T}'$, since 
the subtrees of $c$ and its siblings are the same as in $\str{T}$. We need to see that 
also  the $\cA$-component is retained. 
Let 
$\pi=\type{T'}{c,a_1, \ldots, a_s} \in \prof{\str{T'}}{k}{c}.\cA$. 
Let $\pi'= \type{T}{b,a_1, \ldots, a_s}$.  
As $\pi' \in \prof{\str{T}}{k}{b}.\cA$, by the assumption of the Lemma we have that $\pi' \in \prof{\str{T}}{k}{b'}.\cA$.
Let $b_1, \ldots, b_s$ be elements in position $A$ to $b'$ in $\str{T}$ such that $\pi'=\type{T}{b', b_1, \ldots, b_s}$.
Observe that then  $\type{T}{c,b_1, \ldots, b_s} = \pi$ and thus $\pi \in \prof{\str{T}}{k}{c}.\cA$.

In the opposite direction assume that $\pi=\type{T}{c,a_1, \ldots, a_s} \in \prof{\str{T}}{k}{c}.\cA$. 
Let $\pi' = \type{T}{b',a_1, \ldots, a_s}$.  
As $\pi' \in \prof{\str{T}}{k}{b'}.\cA$, by the assumption of the Lemma we have that $\pi' \in \prof{\str{T}}{k}{b}.\cA$.
Let $b_1, \ldots, b_s$ be elements in position $A$ to $b$ in $\str{T}$ such that $\pi'=\type{T}{b, b_1, \ldots, b_s}$.
Observe that then $\type{T'}{c,b_1, \ldots, b_s} = \pi$ and thus $\pi \in \prof{\str{T'}}{k}{c}.\cA$.

\medskip\noindent
(ii) Consider now the element $a$. Clearly the $\cL$-, $\cA$- and $\cR$-component of $\prof{\str{T}}{k}{a}$ are retained
in $\str{T}'$, since from the point of view of $a$, the only part of the tree which changes is its subtree, and  this change may influence at most the $\cB$-component 
of the profile of $a$. That this component also does not change follows straightforwardly from the assumption of the Lemma that 
$\sigma$-$\prof{\str{T}}{k}{a}.\cB=\sigma$-$\prof{\str{T}}{k}{a'}.\cB$, since in $\str{T}'$ the subtree of $a$ is replaced by the subtree of $a'$ (with the exception
of the root which is still $a$). 

\medskip\noindent
(iii) If $c \not= a$ and $c$ is not a child of $a$ then  note that $c$ retains in $\str{T}'$ all its \emph{direct neighbours} (that is the parent, the children,
the next sibling and the previous sibling) from $\str{T}$. We will use the fact that each component of the profile of an element is determined by its $1$-type (which is obviously
retained from $\str{T}$) and some components of the profiles
of its direct neighbours, as stated in Lemma \ref{l:uniqdeter}. 
We will now systematically analyse the profiles of the elements of $\str{T}'$

(a) If $c$ belongs to the subtree rooted at a child $c'$ of $a$
then  the $\cL$-, $\cB$- and $\cR$-component of $\prof{\str{T}}{k}{c}$ are retained in $\str{T}'$, because the subtrees of $c'$ and its siblings are exactly as
in $\str{T}$.
For the $\cA$-component we proceed by induction on the depth of $c$ in the subtree of $c'$ using the fact that the $\cA$-component
of the profile of an element $d$ is uniquely determined by its $1$-type and by the $\cL$-, $\cA$- and $\cR$-components of the profile of its parent
(Lemma \ref{l:uniqdeter} (i)(a)).

(b)  If $c$ is a left sibling  of $a$ then  
 the $\cL$-, $\cB$- and $\cA$-component of $\prof{\str{T}}{k}{c}$ are retained in $\str{T}'$, because, from the point of view of $a$ only some
elements in position $R$ could change.
For the $\cR$-component we proceed by induction on the distance of $c$ from $a$ using the fact that the $\cR$-component
of the profile of an element $d$ is uniquely determined by its $1$-type and by the $\cR$- and $\cB$-components of the profile of its next sibling
(Lemma \ref{l:uniqdeter} (ii)(b)).

(c)  If $c$ is a right sibling  of $a$ then we proceed symmetrically (using Lemma \ref{l:uniqdeter} (ii)(a)).

(d)  If $c$ is in the subtree of a sibling $b$ of $a$ then we proceed as in (a) using top-down induction on the distance from $b$ to deal with the $\cA$-component.
  
(e) If $c$ is an ancestor of $a$ then the $\cL$-, $\cA$- and $\cR$-component of $\prof{\str{T}}{k}{c}$ are retained in $\str{T}'$ since, from the point of 
view of $c$ only some of its descendants has changed.
For the the $\cB$-component we proceed by induction on the distance of $c$ from $a$ using the fact that the $\cB$-component
of the profile of an element $d$ is uniquely determined by its $1$-type and by the $\cL$-, $\cB$- and $\cR$-components of the profile of its any child, in particular of the child on the vertical path to $a$ (Lemma \ref{l:uniqdeter} (i)(b)).

(f) If $c$ is a sibling of an ancestor of $a$ then we proceed as in (b) or (c)

(g) Any remaining $c$ is now in the subtree rooted at an element about which we already know that its profile is retained from $\str{T}$ and thus we can proceed as in 
(a) and (d)  using top-down induction to deal with the $\cA$-component.
\end{proof}

\begin{lem}\label{l:tcont2}
Let $\sigma$ be a signature with the navigational part containing $\lessh$. Let $\str{T}$ be a $\sigma$-tree and $a,a' \in T$ 
be such that $a'$ is a following sibling of $a$, 
$\sigma$-$\prof{\str{T}}{k}{a}.\cL=\sigma$-$\prof{\str{T}}{k}{a'}.\cL$ and $\sigma$-$\prof{\str{T}}{k}{a}.\cR=\sigma$-$\prof{\str{T}}{k}{a'}.\cR$.
Let $\str{T}'$ be the tree obtained from $\str{T}$ by removing all the subtrees rooted at the 
elements lying on the horizontal path from $a$ to $a'$, including $a$ and excluding $a'$ (and thus making the next sibling of $a'$ in 
$\str{T}$ the next sibling of $a$ in $\str{T}'$).
Then, for any node $c \in T'$ we have $\sigma$-$\prof{\str{T}}{k}{c}=$ $\sigma$-$\prof{\str{T}'}{k}{c}$. 
\end{lem}

\begin{proof}
Let $b$ be the next sibling of $a'$ in $\str{T}$ and $e$ their parent. As in the previous proof we first see that the profiles of $a,b$ and $e$, that is the elements
lying next to the cut, do not change and then propagate our analysis to the remaining elements. 
 
\medskip\noindent
(i) For $a$ it is clear that the $\cL$-, $\cB$- and $\cA$-components of its profile do not change. For the $\cR$-component we naturally use the assumption of the lemma 
that  $\sigma$-$\prof{\str{T}}{k}{a}.\cR=\sigma$-$\prof{\str{T}}{k}{a'}.\cR$ (note that it implies, in particular, that the $1$-types of
$a$ and $a'$ are identical) and the fact that what $a$ can see in position $R$ in $\str{T}'$ is exactly what
$a'$ can see in position $R$ in $\str{T}$. 

\medskip\noindent
(ii) The case of $b$ is symmetric.

\medskip\noindent
(iii) For $e$, the $\cL$-, $\cR$- and $\cA$-components of its profile clearly do not change. For the $\cB$-component take any $\pi \in \prof{\str{T}}{k}{e}.\cB$
and assume $\pi = \type{T}{e, a_1, \ldots, a_s}$, for $a_1, \ldots, a_s$ in position $B$ to $e$. W.l.o.g.~assume that $a_1, \ldots, a_t$ are the elements lying in position $L$ or $B$ to $a$ or being $a$ itself,
and that $a_{t+1}, \ldots, a_t$ are the elements in position $R$ to $a$.  Since $\pi'=\type{T}{a, a_{t+1}, \ldots, a_s}$ belongs to $\prof{\str{T}}{k}{a}.\cR$,
by the assumption of the Lemma we have that $\pi' \in \prof{\str{T}}{k}{a'}.\cR$. Let $b_{t+1}, \ldots, b_s$ be elements in position $R$ to $a'$ in $\str{T}$ 
such that $\pi' = \type{T}{a', b_{t+1}, \ldots, b_s}$. Now $\type{T'}{e, a_1, \ldots, a_t, b_{t+1}, \ldots, b_s} = \pi$ and thus $\pi \in \prof{\str{T}'}{k}{e}.\cB$.
In the opposite direction we proceed similarly.

\medskip\noindent
(iv) for the remaining elements of $\str{T}'$ we argue analogously as in case (iii) of the proof of Lemma \ref{l:tcont1}, that is we use
the fact that those elements retain their neighbours from $\str{T}$ and using Lemma \ref{l:uniqdeter} to propagate the equalities of profiles of computedin  $\str{T}$ and $\str{T}'$ towards the
other parts of $\str{T}'$.
\end{proof}

With the help of the above lemmas  we can now easily get essentially optimal upper bounds on the lengths of paths.

\begin{thm}\label{t:size}
There are a fixed doubly exponential function $\fg$ and a singly exponential function $\ff$ such that:
\begin{enumerate}[(i)]
\item For any navigational signature $\sigma_{nav} \subseteq \{\succv, \lessv, \succv, \lessh \}$ every satisfiable \ODF$[\sigma_{nav}]$ formula $\phi$ has a model in which  horizontal and vertical paths have length bounded from above by $\fg(\sizeOf{\phi})$.
\item Any satisfiable formula $\phi$ in \ODF$[\succv, \lessv, \lessh]$ has a model in which horizontal paths have length bounded from above
by  $\ff(\sizeOf{\phi})$ (and vertical paths are bounded by $\fg(\sizeOf{\phi})$).

\item Any satisfiable formula in \ODF$[ \lessv, \succh, \lessh]$ has a model in which the length of vertical paths 
is bounded from above by $\ff(\sizeOf{\phi})$ (and horizontal paths are bounded by $\fg(\sizeOf{\phi})$).
\item Any satisfiable formula in \ODF$[ \lessv,  \lessh]$ has a model in which vertical paths and 
horizontal paths have length bounded from above by $\ff(\sizeOf{\phi})$. 

\end{enumerate}
\end{thm}

\begin{proof}
Let us take a normal form formula $\phi$ over a signature $\sigma=\sigma_0 \cup \sigma_{nav}$ and denote by $k$ its
width. Let $\str{T} \models \phi$.
First, until there are elements $a,a' \in T$ meeting the assumptions of Lemma \ref{l:tcont2} replace $\str{T}$ by $\str{T}'$ as
in this lemma. Let $\str{T}^*$ be the tree eventually obtained. Clearly every horizontal path in $\str{T}^*$  contains elements of distinct $k$-$\sigma$-profiles. 
By Lemma \ref{l:profilestrees} the number of such profiles is bounded by $\fg^*(|\sigma_0|, k)$. As $k \le |\phi|$ and 
as we may assume that $\sigma_0$ consists only of the unary relations appearing in $\phi$, also $|\sigma_0| \le \sizeOf{\phi}$ we get that $\str{T}^*$ has paths bounded by $\fg^*(\sizeOf{\phi}, \sizeOf{\phi})$, doubly exponentially in $\sizeOf{\phi}$. 

Further, take $\str{T}:=\str{T}^*$ and as long as there are elements $a,a',b,b' \in T$ meeting the assumptions of Lemma \ref{l:tcont1} replace $\str{T}$
by $\str{T}'$ as in this lemma. Let $\str{T}^\dagger$ be the tree eventually obtained. 
Take a vertical path in $\str{T}^\dagger$ and split it into segments consisting of two consecutive elements each
(possibly with the exception of the last segment which may consists of a single element if the number of elements on the path
is odd). Clearly every two pairs have different combination of $k$-$\sigma$-profiles, 
since otherwise a further contraction step would be possible. The number of such combinations is 
bounded by $(\fg^*(\sizeOf{\phi}, \sizeOf{\phi}))^2$, doubly exponentially in $\sizeOf{\phi}$. 
As horizontal paths in $\str{T}^\dagger$ are also horizontal paths in $\str{T}^*$ we have that $\str{T}^\dagger$ is a witness to (i),
where as $\fg(\sizeOf{\phi})$ we take $2\fg^*(\sizeOf{\phi}, \sizeOf{\phi})^2+1$ (two elements in each pair plus possibly the last element on the path if their number is odd).

Next, note that if $\succh \not\in \sigma_{nav}$ then,  if $\str{T}^* \models a \lessh a'$ and $a, a'$ have the same $1$-type then $\sigma$-$\prof{\str{T}^*}{k}{a}.\cL \subseteq \sigma$-$\prof{\str{T}^*}{k}{a'}.\cL$
and $\sigma$-$\prof{\str{T}^*}{k}{a}.\cR \supseteq \sigma$-$\prof{\str{T}^*}{k}{a'}.\cR$. 
Thus, when moving along a horizontal path from left to right through the elements of the same $1$-type, the $\cL$-components of the profiles of elements either stay unchanged or grow, and the $\cR$-components either stay unchanged or diminish, but in each step at least
one of these must change since otherwise a contraction step as in Lemma \ref{l:tcont2} would be possible.
As the number of $1$-types and the size of $\cL$- and $\cR$-components is bounded exponentially in $\sizeOf{\phi}$ (cf. Lemma \ref{l:profilestrees}) we conclude that   the horizontal paths in $\str{T}^*$ are bounded exponentially in $\sizeOf{\phi}$. 
This justifies (ii).

Reasoning similarly as in the above paragraph, but using the $\cA$- and $\cB$-components and Lemma \ref{l:tcont1} we
can show that if $\succv \not\in \sigma_{nav}$ then the vertical paths in $\str{T}^\dagger$ are bounded exponentially in $\phi$.
Take a vertical path in $\str{T}^\dagger$ and split it into segment of size two. 
If there are segments $\langle a, b \rangle$ and $\langle a', b' \rangle$ such that $a'$ is a descendant of $b$,
the $1$-types of $a$ and $a'$ are equal, and the $1$-types of $b$ and $b'$ are equal then 
$\prof{\str{T}^*}{k}{b}.\cA \subseteq \sigma$-$\prof{\str{T}^*}{k}{b'}.\cA$ and
$\prof{\str{T}^*}{k}{a}.\cB \supseteq \sigma$-$\prof{\str{T}^*}{k}{s'}.\cB$,
but at least one of the above inclusions must be strict since otherwise a contraction step as in Lemma \ref{l:tcont1} would
be possible. Since the sizes of the components are bounded exponentially in $\sizeOf{\phi}$, there are exponentially many 
segments for any fixed pair of $1$-types of its elements. As the number of $1$-types is also bounded exponentially 
we get (iii). 

Finally, if none of $\succv$, $\succh$ belongs to $\sigma_{nav}$, then by the arguments above, $\str{T}^\dagger$ has exponentially bounded horizontal
and vertical paths, which proves (iv).
\end{proof}

\subsection{Complexity}

Using Theorem \ref{t:size} one could establish the optimal upper complexity bounds for satisfiability 
of \ODF{} over trees for any navigational signature. We concentrate on the case of the signature
$\{ \lessv, \succh, \lessh \}$  which will allow us to complete the picture concerning the complexity
of satisfiability, and then roughly explain how similar approach
can be used to directly prove the upper bound in Theorem \ref{t:treefull}, which we have already proved
by a reduction to the unary negation fragment.

\begin{thm} \label{t:expspace}
Let $\{ \lessv \} \subseteq \sigma_{nav} \subseteq \{ \lessv, \succh, \lessh \}$. Then the satisfiability problem
for \ODF$[\sigma_{nav}]$ is \ExpSpace-complete. 

\end{thm}
\begin{proof}
The lower bound for \ODF$[\lessv]$ is inherited from \FOt$[\lessv]$, \cite{BBC16}, which in turn refers to
\ExpSpace-hardness of the so-called one-way two-variable guarded fragment, \cite{Kie06}.

To justify the upper bound for \ODF$[\lessv, \succh, \lessh ]$ we propose a nondeterministic algorithm working in exponential space checking if a given
normal form  formula $\phi$ is satisfiable.  As by Savitch theorem \NExpSpace=\ExpSpace{}, the result follows.

Let $k$ be the width of $\phi$. 
Our algorithm attempts to construct a model $\str{T}$ together with  functions $\Omega$ and $\Xi$ 
assigning to  each node $a \in T$ a $1$-type and, respectively, a tuple $(\cF, \cA, \cB, \cL, \cR)$ of sets of $s$-types for various $s \le k$,
intended to be the $1$-type and, respectively, the $k$-profile of $a$ in $\str{T}$. 
For each constructed tuple $(\cF, \cA, \cB, \cL, \cR)$ it immediately checks if $\cF=\tt{fulltype}(\cA, \cB, \cL, \cR)$ and
rejects if it is not the case. 
Additionally, the algorithm stores for each node $a$ its position $hcount(a)$ in the horizontal path of its siblings and
its position $vcount(a)$ on the vertical path from the root to $a$.   

The algorithm starts with constructing the root $\epsilon$ of $\str{T}$, that is by guessing the values $\Omega(\epsilon)$, $\Xi(\epsilon)$ and 
setting $hcount(\epsilon):=0$ and $vcount(\epsilon):=0$. It then verifies that the values of $\Omega$ and $\Xi$ on $\epsilon$ respect
Conditions (a), (c) and (d) of Def.~\ref{d:localcons}. 

Then the algorithm works in a depth-first manner, that is, being at a node $a$ it first goes down to the leftmost child of $a$ (or just decides that $a$ is a leaf), analyses the subtree of
$a$, marks $a$ as "visited", then goes right to the next sibling of $a$, proceeds $a$, and so on;  when it decides that the rightmost child in a horizontal path of siblings is reached
it goes up.

At any moment the algorithm stores the whole vertical path from the root $\epsilon$ to the current node. 
When making a step
down from a node $a$ to a new node $a'$ the algorithm guesses $\Omega(a')$, $\Xi(a')$, sets $hcount(a'):=0$ and $vcount(a'):=vcount(a)+1$,
verifies that the values of $\Omega$ and $\Xi$ on $a$ and $a'$ respect Condition (e) of Def.~\ref{d:localcons}, and that their values
on $a'$ respect condition (c) of Def.~\ref{d:localcons}. 

When making a step
right from a node $a$, which is a child of a node $b$, to a new node $a'$, the algorithm guesses $\Omega(a')$, $\Xi(a')$, sets $hcount(a'):=hcount(a)+1$ and $vcount(a'):=vcount(a)$
and verifies that the values of $\Omega$ and $\Xi$ on $a$ and $a'$ respect Condition (f) of Def.~\ref{d:localcons}, and that their values
on $b$ and $a'$ respect Condition (e) of Def.~\ref{d:localcons}.

When the algorithm nondeterministically decides that the current node  $a$ is a leaf, it verifies that the values of $\Omega$ and $\Xi$ on
$a$ respect Condition (b) of Def.~\ref{d:localcons}. 
When the algorithm nondeterministically decides that the current node $a$ is the rightmost child on a horizontal path, it verifies that the values of $\Omega$ and $\Xi$ on
$a$ respect Condition (d) of Def.~\ref{d:localcons}. 

The algorithm rejects if the value of $hcount$ at any node exceeds $\fg(|\phi|)$ or the value of $vcount$ at any node exceeds $\ff(|\phi|)$
or if the values of $\Xi$ on any node, treated as a profile, is not compatible with $\phi$ (cf.~Lemma \ref{l:tcompatible}). 
It accepts when it returns back to the root without noticing any violation of the local consistency conditions or $\phi$-compatibility. 

That the algorithm uses only exponential space should be clear: it stores a single vertical path of length bounded exponentially by $\ff(|\phi|)$
plus possibly one sibling of the currently inspected node. The values of the counters and the functions $\Omega$ and $\Xi$ stored at each node
are also of exponential size. 

Let us finally explain the correctness of the algorithm. Assume that $\phi$ is satisfiable. By Thm.~\ref{t:size} (iii) $\phi$ has
a model $\str{T}$ with vertical paths bounded by $\ff(|\phi|)$ and horizontal paths bounded by $\fg(|\phi|)$. An accepting run of the algorithm
can be then naturally constructed by making all the guesses in accordance with $\str{T}$. In the opposite direction, if the algorithm has 
an accepting run, then we can naturally extract from this run a tree $\str{T}$ in which the $1$-types of nodes are
as given by the function $\Omega$.
Since during the run the local consistency of the pair $(\Omega, \Xi)$ is checked, it follows by \ref{l:localcons} that for each $a \in T$
we have $\Xi(a)=\prof{T}{k}{a}$. That $\str{T} \models \phi$ follows then by  Lemma \ref{l:tcompatible},
since 
the algorithm verifies at each node $a$ that $\Xi(a)$ is compatible with $\phi$.
\end{proof}
 
As promised, we shall finally briefly explain how  to reprove the upper bound in Thm.~\ref{t:treefull} using the technique we 
have developed.
Recall that Thm.~\ref{t:size} (i) says that every satisfiable formula has a model in which the length of vertical and horizontal paths
is bounded from above doubly exponentially by the function $\fg$. Our examples illustrated in Fig.~\ref{f:longpath} demonstrate that in the 
case of $\sigma_{nav}=\{ \succv, \lessv, \succh, \lessh \}$ we indeed need to take into account models with at least doubly exponential paths,
which thus have triply exponentially many nodes. This means that to fit in \TwoExpTime{} we cannot walk through the whole model.
Instead we propose an algorithm for an alternating machine with exponentially bounded space. This suffices for our purposes,
since by the well known result by Chandra, Kozen and Stockmeyer \cite{CKS81} \AExpSpace=\TwoExpTime, that is any algorithm working
in alternating exponential space can be turned into an algorithm working in doubly exponential time. 

Given a normal form \ODF$[\succv, \lessv, \succh, \lessh]$ formula  $\phi$,
the algorithm attempts to construct a single walk through a tree being a model of $\phi$ from the root to a leaf, at each node making a universal choice whether to go down (to the leftmost child of the current node) or 
to go right (to the next sibling of the current node).

The other details are as in the procedure from the 
the proof of Thm.~\ref{t:expspace}: 
the algorithm operates on similar data structures, that is at each nodes it
guesses the values of $\Omega$ and $\Xi$ and appropriately updates the counters $vcount$ and $hcount$.
At each step it also guesses if the current node is the rightmost child or a leaf,
checks if the values of the counters do not exceed $\fg(\sizeOf{\phi})$,
and verifies the local consistency conditions from \ref{l:localcons} and $\phi$-compatibility conditions from Lemma \ref{l:compatible}.  

Arguments similar to those from the proof of Thm~\ref{t:expspace} ensure that $\phi$ has a model iff the algorithm has an accepting run.

\section{Conclusions} \label{s:conc}

In this paper we investigated the one-dimensional fragment of first-order logic, \ODF{}, over words and trees and collated our
results with the results on a few important formalisms for speaking about those classes of structures. 

Regarding expressivity, all the considered formalisms (\CoreXPath, \GFt, \FOt, \Ct, \UNFO, \ODF) are equiexpressive over words, while over trees it depends on the navigational signature:
over XML trees (child, descendant, next sibling, following sibling) again all the logics are equiexpressive, but over unordered 
trees (only child and descendant) they differ in the expressivity, with \ODF{} being as expressive as
\Ct{} but more expressive than \UNFO{} and \FOt{}. 

\medskip
Concerning the complexity of the satisfiability problem, the picture is presented in Table \ref{tab:comp}.
Column $\{ \succh, \lessh \}$ concerns the case of words. The remaining columns show the results for the case of trees.
We have chosen the four most interesting navigational signatures (XML trees, unordered trees with both child and descendant, and unordered trees
accessible by only descendant or only child). In the case of \CoreXPath{}  we assume that both downward and upward modalities (and both left and right modalities in the case of XML trees)
are present in each of the considered variations

{\small
\begin{table}[h]
\begin{tabular}{c|c|c|c|c|c}
\hline 
&$\{\succh, \lessh \}$&$\{\succv, \lessv, \succh, \lessh \}$ & $\{\succv, \lessv \}$ & $\{\lessv \}$ & $\{\succv \}$\\\hline
\UTL / \CoreXPath &\PSpace & \ExpTime & \ExpTime & {\color{gray} \PSpace} & { \PSpace}\\
\GFt &\NExpTime&\ExpSpace & \ExpSpace & \ExpSpace & \ExpTime \\
\FOt &\NExpTime&\ExpSpace & \ExpSpace & \ExpSpace & \NExpTime \\
\Ct  &\NExpTime& \ExpSpace & \ExpSpace & \ExpSpace &  {\color{gray} \NExpTime} \\
\UNFO &\bf{\NExpTime}& \TwoExpTime & \TwoExpTime & \bf{\ExpSpace} & \TwoExpTime \\
\ODF &\bf{\NExpTime}& \bf{\TwoExpTime} & \bf{\TwoExpTime} & \bf{\ExpSpace} & \bf{\TwoExpTime} \\
\hline 

\end{tabular}

\caption{Complexity over words and trees. Results in bold are proved in this paper. We have not found the results in grey in the literature but they can be easily derived using the existing techniques (see the Appendix).}
\label{tab:comp}
\end{table}
}

For convenience, below we recall the references to the results in the table.
The \PSpace{} result for \UTL{} is proved in \cite{EVW02}. 
For \CoreXPath{} the \ExpTime-results follow from \cite{Mar04}, while the \PSpace-completeness for $\{ \succv \}$ 
is proved in \cite{BFG08}; 
the argument for \PSpace-completeness in the case of $\{ \lessv \}$ is sketched
in the Appendix. \NExpTime-completeness of \FOt{} over words is shown in \cite{EVW02}; this holds also
for \GFt{}, as in the case of words every pair of elements is guarded by $\lessh$ and thus any \FOt{} formula
can be easily translated into \GFt{}.
\GFt{} and \FOt{} over trees are thoroughly examined in \cite{BBC16}.
\NExpTime-completeness of \Ct{} over words is shown in \cite{CW16}.
\Ct{} over trees is investigated in  \cite{BCK17} where \ExpSpace-completeness for signatures containing $\lessv$ is proved;
the signature $\{ \succv \}$ is not studied there, and we sketch an argument for \NExpTime-completeness in this case
in the Appendix. Finally, \TwoExpTime-results for \UNFO{} are proved in \cite{SC13}.

What is probably interesting to note is that in the case of the two-variable logics over trees, it is the signature $\{ \succv \}$ which is easier
than the other signatures. In the case of the multi-variable logics \ODF{} and \UNFO{} this signature is equally hard as our full navigational signature, but a complexity drop can be
observed this time for the signature $\{ \lessv \}$.

\medskip
One more interesting issue that we have not investigated in detail in this paper is succinctness.
Our work implies that \ODF{} is exponentially more succinct than \FOt{} over XML trees and more generally over trees whose signatures contain $\succv$. This follows from the fact that
every satisfiable \FOt$[\succv, \lessv, \succh, \lessh]$ formula has a model whose paths are bounded exponentially in its size
 \cite{BBC16} (this holds also for \Ct{} \cite{BCK17}) while already in \ODF$[\succv]$ we can enforce models with doubly exponentially long paths (Section \ref{s:sizeofmodels} of this paper). 

The above argument does not work in the case of words, since we have shown that in \ODF{} at most exponentially large models
can be enforced, and this indeed can be also done in \FOt{}. Nevertheless, we suspect that also in this case \ODF{} is more succint,
which is suggested by the examples presented in the Introduction.

\section{Funding}
This work was supported by the project ``\emph{Theory of computational logics}''
funded by the Academy of Finland [grants 324435; 328987 to AK] and  the project 
 ``\emph{A quest for new computer logics}'' funded by Polish National Science Centre [grant 2016/21/B/ST6/01444 to EK].

\bibliographystyle{alpha}
\bibliography{mybib}

\appendix
\section{Missing complexities}

\begin{thm}
The satisfiability problem for \CoreXPath{}$[\lessv]$ is \PSpace-complete.
\end{thm}
\begin{proof}(Sketch) 
The lower bound can be shown in a standard fashion, \eg, by a reduction from the QBF problem.

	To get the upper bound we  show that every satisfiable formula has a model of depth and degree bounded polynomially in its length. 
	This can be done by the following standard selection process.
	We extend the language by box modalities $[ \downarrow_+]$, $[\uparrow^+]$, with their standard semantics: $[\cdot] \phi := \neg \langle \cdot \rangle \neg \phi$.
For a given input formula, using de Morgan laws, we push all the negations down to the propositional variables. Let $\phi$ be the NNF result
of this process, and let $SF(\phi)$ be the set of its subformulas.

Take now a tree $\str{T}$ and its node $c$ such that $\str{T}, c \models \phi$. We first take care of the length of vertical path.
For each $\langle \downarrow_+ \rangle \psi \in SF(\phi)$ mark all minimal nodes $a$ such that $\str{T}, a \models \psi$. Similarly,
for each $\langle \uparrow^+ \rangle \psi \in SF(\phi)$ mark all maximal nodes $a$ such that $\str{T}, a \models \psi$. Mark also the
element $c$. Let $\str{T}^*$ be the result of removing from $\str{T}$ all the unmarked elements and rebuilding the structure
of the tree on the marked ones, so that the relation $\lessv$ from $\str{T}$ is respected.
By the structural induction we can now show that for any $\psi \in SF(\phi)$ and any node $a \in T^*$, if
$\str{T}, a \models \psi$ then $\str{T}^*, a \models \psi$; in particular $\str{T}^*, c \models \phi$. 
Since the size of $SF(\phi)$ is linear in $\sizeOf{\phi}$ it follows that the paths of $\str{T}^*$ are also bounded linearly
in $\sizeOf{\phi}$. 

Next we take care of the degree of nodes. Proceeding in breadth-first manner we repeat for all nodes $a$ of $\str{T}^*$:
for every $\langle \downarrow_+ \rangle \psi$, if $\psi$ holds at an descendant of $a$ then mark one such descendant;
mark also $c$ if it is a descendant of $a$. Remove all the subtrees rooted at the children of $a$ which do not contain
any marked node. After this process the resulting tree has the degree of nodes and the length of the vertical paths bounded linearly in $\sizeOf{\phi}$.

Finally, we can check the existence of models with linearly bounded length of paths and degree by guessing their nodes
in a depth-first manner. A natural decision procedure can be designed to work in N\PSpace=\PSpace.
\end{proof}

\begin{thm}
The satisfiability problem for \Ct$[\succv]$ is \NExpTime-complete.
\end{thm}
\begin{proof}(Sketch) The lower bound is inherited from monadic \FOt{} (with no navigational predicates). 
The upper bound can be proved by an adaptation of the  upper bound proof for \FOt$[\succv, \succh, \lessh]$ in
\cite{BBC16}. It will work even in the richer scenario of \Ct$[\succv, \succh, \lessh]$. We first convert the input formula into Scott-type normal form from \cite{BCK17}:

$$\varphi =
  \forall x \forall y \ \chi(x,y) \wedge \bigwedge_{i=1}^{m} \left(
    \forall x \ \exists^{\bowtie_i C_i} y \ \chi_i(x,y) \right),$$
  where $\bowtie_i \in \lbrace \leq, \geq \rbrace$, each $C_i$ is a
  natural number, and $\chi(x,y)$ and all the $\chi_i(x,y)$ are
  quantifier-free. Denote $C=\max \{C_i \}_{i=1, \ldots, m}$.
	
	We then mostly repeat the construction from the proof of  Thm.~4.1 from \cite{BBC16}.
	We start from a model of $\varphi$ with exponentially bounded horizontal and vertical paths as guaranteed by
	Thm.~18 in \cite{BCK17}. Then, the only real modification of the proof from \cite{BBC16} is that when selecting the set of \emph{protected witnesses} $W_1$ we choose 
	$C$ representatives of each $1$-type (or all of them if there are less than $C$ of them) rather than just one.
	Similarly, when selecting \emph{incomparable witnesses} for the elements of $W_1$ we also add to the set $W_2$  $C$ incomparable witnesses 
	for each element and conjunct of type $\forall \exists$ (or all of them if there are less than $C$ of them). 
	Since the number of $1$-types and the value of $C$ are bounded exponentially in the size of $\varphi$ it follows that
	the size of $W=W_1 \cup W_2$ is also bounded exponentially. 
	
	We then proceed as in \cite{BBC16} to prove that there exists a model of $\varphi$ with exponentially many non-isomorphic subtrees.
	Such models can be represented as DAGs of exponential size, which can be then naturally used to test satisfiability in \NExpTime{} by
	guessing such a representation and verifying that it indeed encodes a model of $\varphi$.
	\end{proof}


\end{document}

%% file: macros.tex
\newcommand{\FOt}{\mbox{$\mbox{\rm FO}^2$}}

\newcommand{\ODF}{\mbox{$\mbox{\rm F}_1$}}
\newcommand{\Ct}{$\mbox{\rm C}^2$}
\newcommand{\UTL}{\mbox{\rm UTL}}
\newcommand{\ODW}{\mbox{$\mbox{\rm F}_1[\succh, \lessh]$}}
\newcommand{\OD}{\mbox{$\mbox{\rm F}_1$}}
\newcommand{\UNFO}{\mbox{\rm UNFO}}
\newcommand{\CoreXPath}{\mbox{\rm CoreXPath}}

\newcommand{\GFt}{$\mbox{\rm GF}^2$}
\newcommand{\UF}{\mbox{\rm UF}$_1$}

\newcommand{\PSpace}{\textsc{PSpace}}
\newcommand{\ExpTime}{\textsc{ExpTime}}
\newcommand{\ExpSpace}{\textsc{ExpSpace}}
\newcommand{\NExpSpace}{\textsc{NExpSpace}}
\newcommand{\NExpTime}{\textsc{NExpTime}}
\newcommand{\TwoExpTime}{2\textsc{-ExpTime}}
\newcommand{\TwoNExpTime}{2\textsc{-NExpTime}}

\newcommand{\AExpSpace}{\textsc{AExpSpace}}

\newcommand{\succv}{{\downarrow}}
\newcommand{\lessv}{{\downarrow_{\scriptscriptstyle +}}}
\newcommand{\succh}{{\rightarrow}}
\newcommand{\lessh}{{\rightarrow^{\scriptscriptstyle +}}}

\newcommand{\phie}{\phi^{\sss \exists}}
\newcommand{\phiu}{\phi^{\sss \forall}}
\newcommand{\mse}{m_{\sss \exists}}
\newcommand{\msu}{m_{\sss \forall}}

\newcommand{\cA}{\mathcal{A}}
\newcommand{\cB}{\mathcal{B}}
\newcommand{\cC}{\mathcal{C}}
\newcommand{\cP}{\mathcal{P}}
\newcommand{\cT}{\mathcal{T}}

\newcommand{\cF}{\mathcal{F}}
\newcommand{\cL}{\mathcal{L}}%
\newcommand{\cR}{\mathcal{R}}

\renewcommand{\phi}{\varphi} 

\newcommand{\ff}{\mathfrak{f}}
\newcommand{\fg}{\mathfrak{g}}
\newcommand{\fh}{\mathfrak{h}}

\newcommand{\sizeOf}[1]{\lVert #1 \rVert}
\newcommand{\str}[1]{{\mathfrak{#1}}}

\newcommand{\N}{{\mathbb N}}

\newcommand{\sss}{\scriptscriptstyle}

\newenvironment{romanenumerate}%
{\begin{list}{{\rm (\roman{enumiii})}}{\usecounter{enumiii}
}}{\end{list}}

\newenvironment{alphaenumerate}%
{\begin{list}{{\rm (\alph{enumii})}}{\usecounter{enumii}
}}{\end{list}}

\newcommand{\lcless}{\lambda^{\sss <}}
\newcommand{\lcsucc}{\lambda^{\sss +1}}
\newcommand{\lceq}{\lambda^{\sss =}}

\newcommand{\prof}[3]{{\rm prof}^{\str{#1}}_{#2}({#3})}

\newcommand{\type}[2]{{\rm type}^{\str{#1}}({#2})}